\documentclass[a4paper,10pt]{article}
 \usepackage{amsopn}
 \usepackage{authblk}
 \usepackage{hyperref}
 \usepackage{rotating}
 \usepackage{overpic}
 \usepackage{enumerate}
 \usepackage{varwidth}
 \usepackage{subfig}
\usepackage{booktabs}
\usepackage{algorithm}
\usepackage[noend]{algpseudocode}
\usepackage{amsthm,amsmath,amssymb}
\usepackage{environ,tabularx}
 \usepackage{color} 
\usepackage{lineno}
\usepackage[sort]{cite}

\usepackage{setspace}
\providecommand{\keywords}[1]
{ \medskip  
  \noindent
  \small	
  \textbf{\textit{Keywords---}} #1
}

\newtheorem{proposition}{Proposition}
\newtheorem{theorem}{Theorem}
\newtheorem{definition}{Definition}
\newtheorem{corollary}{Corollary}
\newtheorem{fact}{Fact}

\algblockdefx[Foreach]{Foreach}{Endfor}[1]{\textbf{for each} #1 \textbf{do}}{\textbf{end for}}

\makeatletter
\newcommand{\problemtitle}[1]{\gdef\@problemtitle{#1}}
\newcommand{\probleminput}[1]{\gdef\@probleminput{#1}}
\newcommand{\problemquestion}[1]{\gdef\@problemquestion{#1}}
\NewEnviron{problem}{
  \problemtitle{}\probleminput{}\problemquestion{}
  \BODY
  \par\addvspace{.5\baselineskip}
  \noindent
  \begin{tabularx}{\textwidth}{@{\hspace{\parindent}} l X c}
    \multicolumn{2}{@{\hspace{\parindent}}l}{\@problemtitle} \\
    \textbf{Input:} & \@probleminput \\
    \textbf{Question:} & \@problemquestion
  \end{tabularx}
  \par\addvspace{.5\baselineskip}
}

\newcommand{\revised}[1]{\textcolor{black}{#1}}
\newcommand{\rerevised}[1]{\textcolor{black}{#1}}

\newcommand\sw{{\sf sw}}

\newcommand\vc{{\gamma}}

\newcommand\pos{{\sf PoS}}
\newcommand\pof{{\sf PoF}}

\newcommand\occ{{\sf mcc}}

\newcommand\problemname{{\sf Local $k$-STA }}

\newcommand\agent{p}
\newcommand\agentt{q}

\newcommand\agents{\mathbf{P}}
\newcommand\preference{f}

\newcommand\preferencefamily{\cF}
\newcommand\arrangement{\pi}

\newcommand\utility{U}

\newcommand\cF{\mathcal{F}}
\newcommand\cI{\mathcal{I}}
\newcommand\cO{O}

\newcommand\yes{\textsc{Yes}}

\newcommand\card[1]{|#1|}

\newcommand\bfC{\mathbf{C}}
\newcommand\bfc{\mathbf{c}}
\newcommand\bfV{\mathbf{V}}
\newcommand\bfx{\mathbf{x}}
\newcommand\bfY{\mathbf{Y}}
\newcommand\bfS{\mathbf{S}}

\newcommand\largeval{n}

\newcommand\ETH{{\rm ETH}}

\newtheorem{nestedclaim}{Claim}

\newtheorem{nestedobservation}[nestedclaim]{Observation}

\newenvironment{claimproof}{\begin{proof}\renewcommand{\qedsymbol}{\claimqed}}{\end{proof}\renewcommand{\qedsymbol}{\plainqed}}

\let\plainqed\qedsymbol

\newcommand*{\eqdef}{=\mathrel{\vcenter{\baselineskip0.5ex \lineskiplimit0pt
                     \hbox{\scriptsize.}\hbox{\scriptsize.}}}%
                     }

\title{Hedonic Seat Arrangement Problems\thanks{An extended abstract of this article can be found in \emph{Proceedings of the 19th International Conference on Autonomous Agents and Multiagent Systems (AAMAS 2020), pp.~1777--1779}. This work was partially supported by JSPS KAKENHI Grant Numbers JP18H04091, 
JP18K11168, 
JP18K11169, 
JP19K21537, 
JP20H05967, 
JP21K17707, 
JP21H05852, 
JP22H00513, 
JP23H04388,
JP24H00697, 
JP25K03076, 
JP25K03077, 
and  JST, CRONOS, Japan Grant Number JPMJCS24K2.
}}

\author[1]{Hans L. Bodlaender}
\author[2]{Tesshu Hanaka}
\author[3]{Lars Jaffke}
\author[4]{Hirotaka Ono}
\author[4]{Yota Otachi}
\author[5]{Tom C. van der Zanden}

\affil[1]{Utrecht University, Utrecht, The Netherlands\\ \texttt{h.l.bodlaender@uu.nl}}
\affil[2]{Kyushu University, Fukuoka, Japan\\ \texttt{hanaka@inf.kyushu-u.ac.jp}}

\affil[3]{NHH Norwegian School of Economics, Bergen, Norway\\ \texttt{lars.jaffke@nhh.no}}
\affil[4]{Nagoya University, Nagoya, Japan\\ \texttt{\{ono,otachi\}@nagoya-u.jp}}

\affil[5]{Maastricht University, Maastricht, The Netherlands\\ \texttt{T.vanderZanden@maastrichtuniversity.nl}}

\date{}
\begin{document}



\maketitle              
%

\begin{abstract}

In this paper, we study a variant of hedonic games, called \textsc{Seat Arrangement}. The model is defined by a bijection from agents with preferences for each other to vertices in a graph $G$. The utility of an agent depends on the neighbors assigned in the graph. More precisely, it is the sum over all neighbors of the preferences that the agent has towards the agent assigned to the neighbor. We first consider the price of stability and fairness for different classes of preferences. In particular, we show that there is an instance such that the price of fairness ({\sf PoF}) is unbounded in general. Moreover, we show an upper bound $\tilde{d}(G)$ and an almost tight lower bound $\tilde{d}(G)-1/4$ of {\sf PoF}, where $\tilde{d}(G)$ is the average degree of an input graph. Then we investigate the computational complexity of problems to find certain ``good'' seat arrangements, say  \textsc{Utilitarian Arrangement}, \textsc{Egalitarian Arrangement}, \textsc{Stable Arrangement}, and \textsc{Envy-free Arrangement}. We give dichotomies of computational complexity of four \textsc{Seat Arrangement} problems from the perspective of the maximum order of connected components in an input graph. For the parameterized complexity, \textsc{Utilitarian Arrangement} can be solved in time $n^{O(\gamma)}$, while it cannot be solved in time $f(\gamma)n^{o(\gamma)}$ under ETH, where $n$ is the number of agents and $\gamma$ is the vertex cover number of an input graph. Moreover, we show that \textsc{Egalitarian Arrangement} and \textsc{Envy-free Arrangement} are weakly NP-hard even on graphs of bounded vertex cover number. Finally, we prove that determining whether a stable arrangement can be obtained from a given arrangement by $k$ swaps is W[1]-hard when parameterized by $k+\gamma$, whereas it can be solved in time $n^{O(k)}$.

\end{abstract}

\section{Introduction}

Given a set of $n$ agents with preferences for each other and an $n$-vertex graph, called the {\em seat graph}, we consider assigning each agent to a vertex in the graph.
Each agent has a utility that depends on the neighbors in the graph.
Intuitively, if a neighbor is preferable to the agent, his/her utility is high.
This models several situations such as seat arrangements in classrooms, offices, restaurants, or vehicles.
Here, a vertex corresponds to a seat and an assignment corresponds to a seat arrangement.
If we arrange seats in a classroom, the seat graph is a grid. As another example, if we consider a round table in a restaurant, the seat graph is a cycle.
We name the model \textsc{Seat Arrangement}. 

\textsc{Seat Arrangement} is related to hedonic games \cite{Handbook2016}. If the seat graph in a \textsc{Seat Arrangement} instance is a disjoint union of cliques, then each clique can
be viewed as a potential coalition. Hence an arrangement on that graph naturally corresponds to a coalition forming.\footnote{Because the size of each coalition in a hedonic game with $n$ agents is not designated in general, one may think the size of each clique in a \textsc{Seat Arrangement} instance cannot be determined. However, this can be done by taking
$n$ cliques of size $n$ and $n^2-n$ neutral agents.} \revised{In that sense, this model is a generalization of hedonic games.}



In this paper, we consider the following problems to find four types of {\em desirable} seat arrangements: \textsc{Utilitarian Arrangement} ({\sf UTA}), \textsc{Egalitarian Arrangement} ({\sf EGA}), \textsc{Stable Arrangement} ({\sf STA}), and \textsc{Envy-free Arrangement} ({\sf EFA}).
{\sf UTA} is the problem to find a seat arrangement that maximizes the sum of utilities of agents, which is called the {\em social welfare}.

The concept of {\sf UTA} is a macroscopic optimality, and hence it may ignore individual utilities.
Complementarily, {\sf EGA} is the problem to find a seat arrangement that maximizes the least utility of an agent.
From the viewpoint of economics, the maximin utility of an arrangement can be interpreted as a measure of fairness~\cite{Bertsimas2011,Caragiannis2012,KerkmannNRRRSW22}.



Stability is one of the central topics in the field of hedonic games including \textsc{Stable Matching}~\cite{GS1962,Noam2007,Aziz2013,Handbook2016}.
Motivated by this, we define a {\em stable} arrangement as an arrangement with no pair of agents that has an incentive of swapping their seats (i.e., vertices), called a {\em blocking pair}.
This corresponds to the definition of {\em exchange-stability} proposed by Alcalde in the context of stable matchings \cite{Alcalde1994}.
In \textsc{Seat Arrangement}, {\sf STA} is the problem of deciding whether there is a stable arrangement in a graph. 

Finally, we consider the envy-freeness of \textsc{Seat Arrangement}.
The envy-freeness is also a natural and well-considered concept in hedonic games.
In \textsc{Seat Arrangement}, an agent $p$ {\em envies} another agent $q$ if $p$ has an incentive of swapping its seat with $q$.
Note that $q$ may not have an incentive of swapping its seat with $p$.
By definition, any envy-free arrangement is stable.


\begin{table}[t]
\centering
\caption{The existence of a stable arrangement and the price of stability for \textsc{Seat Arrangement}~\cite{Massand2019}.}
\label{result:stable}
\begin{tabular}{lllll}
\cline{2-3}
\multicolumn{1}{l|}{}            & \multicolumn{1}{c|}{General}       & \multicolumn{1}{c|}{Symmetric}              &  &  \\ \cline{1-3}
\multicolumn{1}{|c|}{Weighted}   & \multicolumn{1}{c|}{-}             & \multicolumn{1}{c|}{always exists, $\pos = 1$} &  &  \\ \cline{1-3}
\multicolumn{1}{|c|}{Unweighted} & \multicolumn{1}{c|}{may not exist} & \multicolumn{1}{c|}{-}                      &  &  \\ \cline{1-3}
\end{tabular}
\end{table}

\begin{table}[t]
\centering
\caption{The price of fairness for \textsc{Seat Arrangement}.}
\label{result:pof}
\begin{tabular}{c|c|c|}
\cline{2-3}
                             & General  & Binary                 \\ \hline
\multicolumn{1}{|l|}{$\pof$} & $\infty$ & $\theta(\tilde{d}(G))$ \\ \hline
\end{tabular}
\end{table}

\subsection{Our contribution}
In this paper, we first investigate the  price of stability (\pos) and the  price of fairness (\pof) of \textsc{Seat Arrangement}, which are defined as  the ratio of the maximum social welfare over the social welfare of a maximum stable arrangement and a maximin arrangement, respectively.
For the price of stability, it is shown that $\pos$ is 1 under symmetric preferences by a result in \cite{Massand2019}.
For the price of fairness, we show that there is a family of instances such that $\pof$ is unbounded.
For the binary case, we show an upper bound of $\tilde{d}(G)$ of $\pof$, where $\tilde{d}(G)$ is the average degree of the seat graph $G$.
On the other hand, we present an almost tight lower bound $\tilde{d}(G)-1/4$ of $\pof$.
Furthermore, we give a lower bound $\tilde{d}(G)/2+1/12$ for the cases with symmetric preferences.
\revised{Tables \ref{result:stable} and \ref{result:pof} show the price of stability and the  price of fairness of \textsc{Seat Arrangement}.}


Next, we give dichotomies of the computational complexity of four \textsc{Seat Arrangement} problems from the perspective of the maximum order of connected components in the seat graph.
\revised{For  {\sf UTA} and {\sf EGA}, we show that they are solvable in polynomial time if the order of each connected component in the seat graph is at most 2 whereas they are NP-hard even if the order of each connected component of the seat graph is 3.}
Since it can be shown that an arrangement with maximum social welfare is always stable under symmetric preferences, {\sf STA} under symmetric preferences can also be solved in polynomial time if the order of each connected component is at most 2. 
\revised{However, the problem is NP-complete on such graphs without the symmetry property due to the NP-completeness of \textsc{Exchange Stable Roommates}.}
Note that if each connected component in the seat graph is of order at most 1, it consists of only isolated vertices, and hence {\sf STA} is trivially solvable.
\revised{For {\sf EFA}, we show that it can be solved in polynomial time if the order of each connected component in the seat graph is at most 2 \emph{and} the preferences are symmetric or positive. Moreover, it can be shown that the problem is solvable if the order of each connected component in the seat graph is exactly 2. \rerevised{Recently}, Ceylan, Chen, and Roy \rerevised{have shown} that {\sf EFA} remains NP-hard even if the order of each connected component is at most 2~\cite{SA:Ceylan0R23}, which shows the boundary of the complexity of {\sf EFA}.}
We also show the NP-hardness of {\sf UTA}, {\sf EGA}, and {\sf EFA} on very restricted graph classes.
\revised{The complexity results are summarized in Table~\ref{result:comp}.}

\revised{To investigate the complexity more deeply, we then study the parameterized complexity of \textsc{Seat Arrangement}. As we mentioned above, four \textsc{Seat Arrangement} problems are NP-hard on trees, cycles, or forests, which are bounded treewidth graphs. Thus, we consider designing an algorithm parameterized by a more restrictive parameter: vertex cover number.
For {\sf UTA}, we show that it can be solved in time $n^{O(\gamma)}$ where $\gamma$ is the vertex cover number of the seat graph, which means that it is solvable in polynomial time on bounded vertex cover number graphs. On the other hand,  it can be shown that {\sf UTA} is  W[1]-hardness for $\gamma$ and
cannot be solved in time $n^{o(n)}$ and $f(\gamma)n^{o(\gamma)}$ under Exponential Time Hypothesis (\ETH\@). For {\sf EGA} and {\sf EFA}, we prove that they are NP-hard even on seat graphs with $\gamma=2$.}

Finally, we study the parameterized complexity of the local search to find a stable arrangement.
We show that determining whether a stable arrangement can be obtained from a given arrangement by $k$ swaps is W[1]-hard when parameterized by $k+\gamma$, whereas it can be solved in time $n^{O(k)}$.
\revised{Table~\ref{result:param} shows the parameterized complexity of \textsc{Seat Arrangement}.}

\begin{table}[t]
\centering
\caption{The computational complexity of \textsc{Seat Arrangement}. {\sf mcc} stands for the maximum order of connected components in the seat graph. ${\sf ecc}=x$ means that the order of each connected component is exactly $x$.}
\begin{tabular}{l|c|c|}
\cline{2-3}
& Algorithms  & Hardness  \\ \hline
\multicolumn{1}{|c|}{{\sf UTA}} & \begin{tabular}[c]{@{}c@{}}P\\ ($\occ \le 2$)\end{tabular}        & \begin{tabular}[c]{@{}c@{}}NP-hard\\ (${\sf ecc} =  3$, trees, cycles, cluster graphs)\end{tabular} \\ \hline

\multicolumn{1}{|c|}{{\sf EGA}} & \begin{tabular}[c]{@{}c@{}}P\\ ($\occ\le 2$)\end{tabular}        & \begin{tabular}[c]{@{}c@{}}NP-hard\\ (${\sf ecc} =  3$, cycles, cluster graphs)\end{tabular} \\ \hline

\multicolumn{1}{|c|}{{\sf EFA}} & \begin{tabular}[c]{@{}c@{}}P\\ ($\occ\le 2 \land (\text{positive}\lor \text{symm.})$, ${\sf ecc} = 2$)\end{tabular} & \begin{tabular}[c]{@{}c@{}}NP-hard\\ ($\occ \le  2$, trees, cluster graphs)\end{tabular} \\ \hline

\multicolumn{1}{|c|}{{\sf STA}} &                                                                & \begin{tabular}[c]{@{}c@{}}NP-hard\\ (${\sf ecc} = 2$)\end{tabular} \\ \hline
\end{tabular}
\label{result:comp}
\end{table}

\begin{table}[H]
\centering
\caption{The parameterized complexity of \textsc{Seat Arrangement} where $\gamma$     is the vertex cover number of the seat graph.}
\begin{tabular}{l|c|c|}
\cline{2-3}
& Algorithms      & Hardness                                                          \\ \hline
\multicolumn{1}{|l|}{{\sf UTA}   }                                                      & $n^{O(\gamma)}$ & W[1]-hard for $\gamma$                                        \\ \hline
\multicolumn{1}{|l|}{{\sf EGA} }                                                      &         & NP-hard for $\gamma=2$                                            \\ \hline
\multicolumn{1}{|l|}{{\sf EFA}   }                                                    &         & NP-hard for  $\gamma=2$ \\ \hline
\multicolumn{1}{|l|}{Symmetric {\sf STA}}                  & $n^{O(\gamma)}$  &  PLS-complete~\cite{Massand2019}                                                                 \\ \hline
\multicolumn{1}{|l|}{\textsc{\problemname}}  & $n^{O(k)}$      & W[1]-hard for $k+\gamma$                                      \\ \hline
\end{tabular}
\label{result:param}
\end{table}

\subsection{Related work}\label{subsec:motivation}
A {\em hedonic game} is a non-transferable utility game regarding coalition forming, where each agent's utility depends on the identity of the other agents in the same coalition~\cite{Dreze1980,Bogomolnaia2002}.
It includes the \textsc{Stable Matching} problem~\cite{GS1962,GI1989,Manlove2013,Handbook2016}. 
\textsc{Seat Arrangement} can be considered a hedonic game of arrangement on a graph.

Several graph-based variants of hedonic games have been proposed in the literature, see e.g.~\cite{Aziz2013,Branzei2009,Branzei2011,GairingS19,Igarashi2016,Okubo2019}. 
However, they typically utilize graphs to define the preferences of agents, and both the preferences and coalitions define the utilities of agents. 
On the other hand, in \textsc{Seat Arrangement}, the preferences are defined independently of a graph and the utility of an agent is determined by an arrangement in a graph (more precisely, the preferences for the assigned neighbors in the graph).

A major direction of research about hedonic games is the computational complexity of the problems to find desirable solutions such as a solution with maximum social welfare and a stable solution~\cite{Bogomolnaia2002,Aziz2013}.
\revised{Peters \cite{Peters2016} and Hanaka et al.\cite{Hanaka2019,HM2022,HanakaIO25}  investigate the parameterized complexity of various types of hedonic games for several graph parameters (e.g., treewidth).}
For the local search complexity, Gairing and Savani study the PLS-completeness of finding a stable solution~\cite{GairingS19}.
In terms of mechanism design and algorithmic game theory, many researchers study the price of anarchy, the price of stability, and the price of fairness~\cite{Noam2007,Bertsimas2011,Handbook2016,Bilo2018}.

Possibly the closest relative of \textsc{Seat Arrangement} among hedonic games is \textsc{Stable Matching}. In this problem, agents are partitioned into pairs under preferences~\cite{GS1962,Noam2007,Handbook2016}. 
One of the most famous \textsc{Stable Matching} variants is \textsc{Stable Marriage}.
The Gale-Shapley algorithm solves \textsc{Stable Marriage} instances in the standard setting in polynomial time. So far, many variants of \textsc{Stable Matching} have been considered in the literature, including \textsc{Stable Marriage} 
 and \textsc{Stable Roommate}~\cite{GS1962,Irving1994,IRVING1985,Irving2002,Handbook2016}.
 
One might think that \textsc{Stable Roommate} could be modeled as \textsc{Seat Arrangement} on a graph in which each of the connected components is an edge. 
However, these two models are slightly different due to the difference in the definitions of blocking pairs.
A blocking pair in \textsc{Stable Roommate} can deviate from their partners, and then they match each other, whereas they can only swap their seats in \textsc{Seat Arrangement}.
Alcalde proposes the {\em exchange-stability} \cite{Alcalde1994}.
Under the exchange stability, a blocking pair does not deviate from their partner, but they swap with each other.
Cechl\'arov\'a and Manlove \rerevised{prove} that \textsc{Stable Marriage} and \textsc{Stable Roommates} are NP-complete under exchange stability even if the preference lists are complete and strict~\cite{Cechlarova2002,Cechlarova2005}. From this, it follows that {\sf STA} is NP-complete.
Several other papers about exchange-stable matching models are \cite{Irving2008,Aziz2017,GLW2017,Saffidine2018}.

\textsc{Schelling Games} also has a similar model to ours. \textsc{Schelling Games} 
have received
a lot of attention and are well-studied in the field of Artificial Intelligence and Economics. 
They were originally introduced by Schelling to analyze the dynamics of segregation~\cite{SG:Schelling1969,SG:Schelling1971}.
Agarwal et al. \rerevised{propose} \textsc{Schelling Games} on graphs. The game is as follows. We are given a graph $G=(V,E)$ and  $k$ sets of agents $T_1,\ldots,T_k$ where $\mathbf{P}=\{1,\ldots,n\}=\bigcup_i T_i$ and $T_i\cap T_j=\emptyset$ for $i\neq j$.
The sets $T_1,\ldots,T_k$ indicate the \emph{types} of agents. In a sense, we can assume
that the friends of an agent in $T_i$ belong to $T_i$.
The utility of an agent $p\in T_i$ is defined by an assignment $\pi: \mathbf{P}\to V$ as $U_p(\pi)=\sum_{v\in N_G(\pi(p))} |N_G(\pi(p))\cap (T_i\setminus \{p\})|/|N_G(\pi(p))|$, where $N_G(v)$ is the set of neighbors of $v$ in $G$, that is, the utility is the ratio of friends among the neighbors in $G$.\footnote{If a graph $G$ is regular, that is, the degree of each vertex is the same, \textsc{Schelling Games} can be regarded as \textsc{Seat Arrangement} where the preference graph is the union of unweighted disjoint cliques corresponding to $T_1,\ldots,T_k$.} This types of utility is called \emph{fractional}. Several computational results and the efficiency of the game are studied~\cite{SG:AgarwalEGISV21,SG:ChauhanLM18,SG:BiloBLM22,KreiselBFN24}.

In the context of one-sided markets, Massand and Simon considered the problem of allocating indivisible objects to a set of rational agents where each agent’s final utility depends on the intrinsic valuation of the allocated item as well as the allocation within the agent’s local neighborhood \cite{Massand2019}.
Although the problem is motivated from different contexts, it has a quite similar nature to \textsc{Seat Arrangement}, and they also consider stable and envy-free allocation on the problem.
In fact, the following results about \textsc{Seat Arrangement} are immediately obtained from \cite{Massand2019}: \emph{(1) The $\pos$ of  \textsc{Seat Arrangement} is 1. (2) There is an instance of binary \textsc{Seat Arrangement} with no stable arrangement. (3) The local search problem to find a stable solution under symmetric preferences by swapping two agents iteratively is PLS-complete. (4) {\sf EFA} is NP-complete.}
In this paper, we give further and deeper analyses of \textsc{Seat Arrangement}.

After the appearance of the conference version of this paper, several additional 
papers that study \textsc{Seat Arrangement} appeared~\cite{SA:BerriaudCW23,SA:Ceylan0R23,SA:ChenCJS21,SA:Wilczynski23}.

\section{Preliminaries}
We use standard terminologies in graph theory.
Let $G=(V,E)$ be a graph where $n=|V|$ and $m=|E|$.
For a directed graph $G$, we denote the set of in-neighbors (resp., out-neighbors) of $v$ by $N^{\rm in}_G(v)$ (resp., $N^{\rm out}_G(v)$) and the {\em in-degree} (resp., {\em out-degree}) of $v$  by $d^{\rm in}_G(v):=|N^{\rm in}(v)|$ (resp., $d^{\rm out}_G(v):=|N^{\rm out}(v)|$).
For an undirected graph $G$, we denote the set of neighbors of $v$ by $N_G(v)$ and the \emph{degree} of $v$ by $d_G(v)=|N_G(v)|$.
We also define $\Delta(G)=\max_{v\in V}d_G(v)$ and $\tilde{d}(G)=2m/n$ as the maximum degree and the average degree of $G$, respectively.
A {\em vertex cover} $X$ is the set of vertices such that for every edge, at least one endpoint is in $X$. The {\em vertex cover number} of $G$, denoted by $\vc(G)$, is the size of a minimum vertex cover in $G$.
For simplicity, we sometimes drop the subscript of $G$ if it is clear.
For the basic definition of parameterized complexity such as the classes FPT, XP, and W[1], we refer the reader to the book \cite{Cygan2015}.

\subsection{Models}
We denote by $\agents$ the set of agents, and define an {\em arrangement} as follows.
\begin{definition}[Arrangement]\label{def:seat} 
For a set of \revised{$n$} agents $\agents$ and an undirected graph \revised{$G=(V,E)$}, a bijection $\pi: \agents \rightarrow V(G)$ is called an {\em arrangement} in $G$.
\end{definition}
We denote by $\Pi$ the set of all arrangements in $G$. Note that $|\Pi|=n!$. 
We call graph $G$ the {\em seat graph}. 
The definition means that an arrangement assigns each agent to a vertex in $G$.
When we specify that the seat graph $G$ is in some graph class ${\mathcal G}$, we sometimes use the term \textsc{Seat Arrangement} on ${\mathcal G}$.
Moreover, we define the $(p,q)$-swap arrangement for $\pi$.
\begin{definition}[$(p,q)$-swap arrangement]\label{def:pq-seat} 
Let $\agents$ be a set of agents, $G$ be a graph and $\pi$ be an arrangement. For
a pair of agents $p, q \in \agents$, we say that $\pi'$ is the \emph{$(p, q)$-swap arrangement} if $\pi'$ can obtained from $\pi$ by swapping the arrangement of $p$ and $q$, that is, $\pi'$ satisfies that $\pi'(p)=\pi(q)$, $\pi'(q)=\pi(p)$, and $\pi'(r)=\pi(r)$ for every $r\in \agents \setminus \{p,q\}$.
\end{definition}

Next, we define the {\em preference} of an agent.
\begin{definition}[Preference]\label{def:preference}
 The {\em  preference} of $p\in \agents$ is defined by $f_p: \agents \setminus \{p\}\rightarrow \mathbb{R}$. 
\end{definition}
We denote by $\preferencefamily_\agents$ the set of preferences of all agents in $\agents$.
\revised{Here, we say the preferences are {\em binary} if $f_p: \agents \setminus \{p\}\rightarrow \{0,1\}$ for every agent $p$~\cite{SA:Ceylan0R23,Massand2019,OlsenBT12}, are {\em nonnegative} if $f_p: \agents \setminus \{p\}\rightarrow \mathbb{R}^+_{0}$~\cite{Handbook2016,SA:Ceylan0R23,Massand2019}, and are {\em positive} if $f_p: \agents \setminus \{p\}\rightarrow \mathbb{R}^+$~\cite{SA:Ceylan0R23,OlsenBT12}.
Furthermore, we say they are {\em symmetric} if $f_p(q)=f_q(p)$ holds for any pair of agents $p,q\in \agents$~\cite{Aziz2017,SA:Ceylan0R23,GairingS19,OlsenBT12} and {\em strict} if for any $p\in \agents$ there is no pair of distinct $q,r\in \agents$ such that $f_p(q)=f_p(r)$~\cite{Handbook2016,Cechlarova2005,SA:Ceylan0R23,SA:Wilczynski23}.}
The directed and weighted graph $G_{\preferencefamily_{\agents}}=(\agents, E_{\preferencefamily_{\agents}})$  associated with the preferences $\preferencefamily_{\agents}$ is called the {\em preference graph}, where $E_{\preferencefamily_{\agents}}=\{(p,q)\mid f_p(q)\neq 0\}$ and the weight of $(p,q)$ is $f_p(q)$. If the preferences are symmetric, we define the corresponding preference graph as an undirected graph.

Finally, we define the {\em utility} of an agent and the {\em social welfare} of an arrangement $\pi$.
\begin{definition}[Utility and social welfare]\label{def:utility}
 Given an arrangement $\pi$ and the preference of $p$,  the {\em utility} of $p$ is defined by $U_p(\pi)=\sum_{v\in N(\pi(p))} f_p(\pi^{-1}(v))$.
Moreover, the {\em social welfare} of $\pi$ for $\agents$ is defined by the sum of all utilities of agents and denoted by $\sw(\pi) = \sum_{p\in \agents} U_p(\pi)$.
\end{definition}
The function $U_p(\pi)=\sum_{v\in N(\pi(p))} f_p(\pi^{-1}(v))$ represents the sum over all neighbors of the preferences that the agent has towards the agent assigned to the neighbor. 
This function is often used in coalition formation games~\cite{Branzei2009,GairingS19,Aziz2013}.
By definition, if the seat graph is a complete graph, all the arrangements have the same social welfare. 


In the following, we define \revised{five} types of \textsc{Seat Arrangement} problems.
First, we define \textsc{Utilitarian Arrangement}. An arrangement  $\pi^*$ is {\em maximum} if it satisfies $\sw(\pi^*)\ge \sw(\pi)$ for any arrangement $\pi$.
Then,  \textsc{Utilitarian Arrangement} ({\sf UTA}) is defined as follows.
\begin{description}
\item[Input:] A graph $G = (V,E)$, a set of agents $\agents$, and the preferences of agents $\preferencefamily_\agents$.
\item[Task:] Find a maximum arrangement in $G$.
\end{description}

An arrangement $\pi^*$ is called a {\em maximin} arrangement if $\pi^*$ satisfies  $\min_{p\in \agents} U_p(\pi^*)\ge \min_{p\in \agents}  U_p(\pi)$ for any arrangement $\pi$.
Then, \textsc{Egalitarian Arrangement} ({\sf EGA}) is defined as follows.
\begin{description}
\item[Input:] A graph $G = (V,E)$, a set of agents $\agents$, and the preferences of agents $\preferencefamily_\agents$.
\item[Task:] Find a maximin arrangement in $G$.
\end{description}

\revised{We then} define the \emph{stability} of \textsc{Seat Arrangement}.
\begin{definition}[\textsc{Stability}]
Given an arrangement $\pi$, a pair of agents $p$ and $q$ is called a {\em blocking pair} for $\pi$ if it satisfies that $U_p(\pi')> U_p(\pi)$ and $U_q(\pi')> U_q(\pi)$ where $\pi'$ is the $(p,q)$-swap arrangement for $\pi$.
If there is no blocking pair in an arrangement, it is said to be {\em stable}.
\end{definition}
Then, the \textsc{Stable Arrangement} ({\sf STA}) problem is as follows.
\begin{description}
\item[Input:] A graph $G = (V,E)$, a set of agents $\agents$, and the preferences of agents $\preferencefamily_\agents$.
\item[Task:] Decide whether there \rerevised{exists} a stable arrangement in $G$.
\end{description}



Finally, we define the {\em envy-freeness} of \textsc{Seat Arrangement}.
\begin{definition}[\textsc{Envy-free}]
An arrangement $\pi$ is envy-free if there is no agent $p$ such that there exists $q\in \agents\setminus \{p\}$ that satisfies  $U_p(\pi')> U_p(\pi)$ where $\pi'$ is the  $(p,q)$-swap arrangement for $\pi$.
\end{definition}
By the definition of envy-freeness, we have:
\begin{proposition}
If an arrangement is envy-free, it is also stable.
\end{proposition}
Then, \textsc{Envy-free Arrangement} ({\sf EFA)} is defined as follows.
\begin{description}
\item[Input:] A graph $G = (V,E)$, a set of agents $\agents$, and the preferences of agents $\preferencefamily_\agents$.
\item[Task:] Decide whether there \rerevised{exists} an envy-free arrangement in $G$.
\end{description}

In this paper, we assume that the number of agents equals the number of vertices in the seat graph, that is,  $|\agents|=|V(G)|$. However, it is natural to consider a more general model that admits vacant seats. In this case, we can reduce it to the case of $|\agents|=|V(G)|$ by adding $|V(G)|-|\agents|$ \emph{neutral} agents, which stand for vacant seats. Here, a neutral agent $p$ satisfies $f_p(q) = f_q(p) = 0$ for any $q\in \agents\setminus \{p\}$. 
This reduction naturally expands the model except for the stability. In the current definition of stability, even if there exists an agent who wants to move to a vacant seat, he/she cannot move because the neutral agent standing for the vacant seat does not have the incentive to move. Thus, we also consider another stability concept, called \emph{strict stability}.

\begin{definition}[\textsc{Strict stability}]
Given an arrangement $\pi$, a pair of agents $p$ and $q$ is called a {\em weakly blocking pair} for $\pi$ if it satisfies that $U_p(\pi')> U_p(\pi)$ and \rerevised{$U_q(\pi')\ge U_q(\pi)$} where $\pi'$ is the $(p,q)$-swap arrangement for $\pi$.
If there is no blocking pair in an arrangement, it is said to be {\em strictly stable}.
\end{definition}
Then, the \textsc{Strictly Stable Arrangement} ({\sf StrictSTA}) problem is as follows.
\begin{description}
\item[Input:] A graph $G = (V,E)$, a set of agents $\agents$, and the preferences of agents $\preferencefamily_\agents$.
\item[Task:] Decide whether there \rerevised{exists} a strictly stable arrangement in $G$.
\end{description}
It is easily seen that {\sf StrictSTA} is PLS-complete by the same reduction of  {\sf StrictSTA}~\cite{Massand2019}.

\subsection{Tools}
In this subsection, we introduce a polynomial-time algorithm for computing a \emph{maximum weight $k$-matching}, which is a maximum weight matching of size exactly $k$.

\begin{theorem}[Folklore]\label{thm:k-matching}
Let $G=(V,E)$ be a complete graph with arbitrary edge weights. Given an integer $k<\lfloor n/2\rfloor$, a maximum weight $k$-matching of $G$ can be computed in polynomial time.
\end{theorem}
\begin{proof}
    We add $n-2k$ universal vertices adjacent to all the vertices in $V$ with edges of weight $0$.
    We denote the constructed graph by $G'$. Then the number of vertices in $G'$ is even, which implies that $G'$ has a perfect matching. For any perfect matching $M'$ of maximum weight in $G'$, $M'\cap E$ forms a maximum weight $k$-matching in $G$. Since a maximum weight perfect matching with arbitrary edge weights can be computed in polynomial time~\cite{edmonds1965maximum}, one can compute a maximum weight $k$-matching of $G$ in polynomial time.
\end{proof}

\color{black}
\section{Stability and Fairness}\label{sec:stability}
In this section, we study the stability and the fairness of \textsc{Seat Arrangement}.
Let $\Pi_s$ be the set of stable solutions and $\pi^*$ be a maximum arrangement.
Then, the price of stability (\pos) is defined by $\min_{\pi_s \in \Pi_s}  \sw(\pi^*)/\sw(\pi_s)$~\cite{Noam2007}. 
In other words, the price of stability is defined as the gap between the maximum social welfare and the social welfare of a maximum stable solution.
From \cite{Massand2019}, we immediately obtain the following proposition.
 \begin{proposition}[\cite{Massand2019}]\label{prop:stability:Massand}
In symmetric \textsc{Seat Arrangement}, a maximum arrangement is stable, and thus, the $\pos$ of symmetric \textsc{Seat Arrangement} is 1. 
On the other hand, there is an instance of binary \textsc{Seat Arrangement} with no stable arrangement.
\end{proposition}
Proposition \ref{prop:stability:Massand} is easily shown by using the potential function $\Phi(\pi)=\sw(\pi)$.
On the other hand, there is an instance with no envy-free arrangement even if the preferences are symmetric. Consider three agents $x,y,z$ such that $f_p(q)=f_q(p)=1$ for $p,q\in \{x,y,z\}$. If we assign them to a path $P_3$, two agents assigned to endpoints of $P_3$ envy the agent assigned to the center in $P_3$.
\begin{proposition}\label{prop:no_envy_free}
There is an instance of binary and symmetric \textsc{Seat Arrangement} with no envy-free arrangement.
\end{proposition}

Next, we consider the price of fairness~\cite{Bertsimas2011,Caragiannis2012,Handbook2016}.
 Let $\Pi_f$ be the set of maximin solutions.
Then, the price of fairness is defined by $\min_{\pi_f \in \Pi_f}  \sw(\pi^*)/\sw(\pi_{f})$, that is, the ratio between the maximum social welfare and the social welfare of a maximin arrangement.

 \begin{proposition}
There is an instance such that the $\pof$ of \textsc{Seat Arrangement} is unbounded.
\end{proposition}
\begin{proof}
Let  $x\ge y \ge 1$ be two integers and the seat graph $G$ be a graph consisting of two edges.
Finally, we set the preferences of four agents $p_1,p_2,p_3,p_4$ as follows:
$f_{p_1}(p_3)=f_{p_2}(p_4)=f_{p_3}(p_2)=f_{p_4}(p_1)=x$, $f_{p_1}(p_2)=f_{p_2}(p_1)=f_{p_3}(p_4)=f_{p_4}(p_3)=y$, and $f_{p_1}(p_4)=f_{p_2}(p_3)=f_{p_3}(p_1)=f_{p_4}(p_2)=0$.

Since the seat graph consists of two edges, the instance has only three arrangements by symmetry.
If $p_1, p_2$ (resp., $p_3, p_4$) are assigned to the same edge, the social welfare is $4y$ and each agent has the utility $y$.
On the other hand, if $p_1, p_3$ (resp., $p_2,p_4$) or $p_1, p_4$ (resp., $p_2,p_3$) are assigned to the same edge, the social welfare is $2x$ and the least utility is $0$.
If $x$ is an arbitrarily large integer and $y=1$, the social welfare of a maximum arrangement is $2x$.
Then, we have $\pof=2x/4y=x/2y$, and hence $\pof$ is unbounded.
\end{proof}

For the binary case, the $\pof$ is bounded by the average degree $\tilde{d}(G)$ of the seat graph $G$. If the least utility is $0$ for every arrangement, we can choose an arrangement with maximum social welfare. In this case, $\pof$ is $1$. Otherwise, the least utility is $1$ and the social welfare of such an arrangement is at least $n$. Since the social welfare is at most $2m$, $\pof$ is bounded by $2m/n=\tilde{d}(G)$.
 \begin{proposition}
For any instance, the $\pof$ of binary \textsc{Seat Arrangement} is at most $\tilde{d}(G)$.
\end{proposition}

\begin{figure}
	\centering
	\includegraphics[height=.16\textheight]{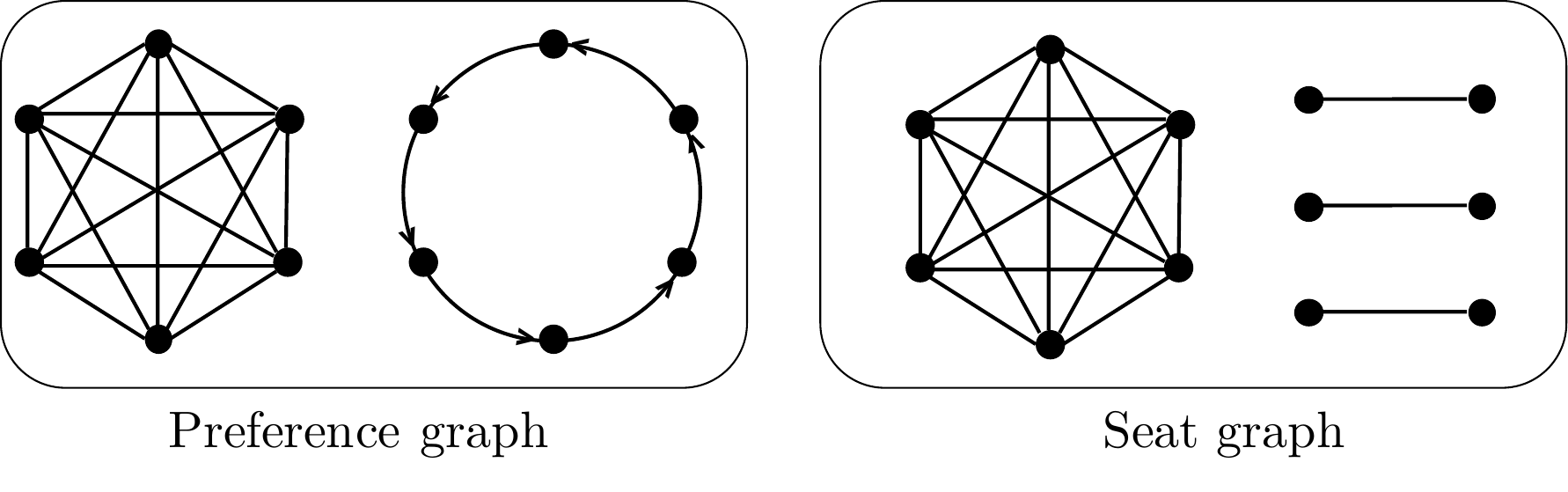}
	\caption{The preference graph (left) and the seat graph (right) in the proof of Proposition~\ref{prop:pof}. }
	\label{fig:PoF}
\end{figure}	
	
Finally, we give almost tight lower bounds of the price of fairness for binary \textsc{Seat Arrangement}.
 \begin{proposition}\label{prop:pof}
There is an instance such that the $\pof$ of binary \textsc{Seat Arrangement} is at least $\tilde{d}(G)-1/4$. Furthermore, there is an instance such that the $\pof$ of binary and symmetric \textsc{Seat Arrangement} is at least $\tilde{d}(G)/2+1/12$.
\end{proposition}
\begin{proof}
We construct such an instance.
Let $\agents_K$ and $\agents_C$ be sets of agents each having $n$ members \revised{where $n$ is even}.
The preference graph consists of an undirected clique and a directed cycle.
The seat graph $G$ consists of a clique of size $n$ and $n/2$ disjoint edges.
The number of edges in $G$ is $n(n-1)/2+n/2=n^2/2$ and the average degree of $G$ is $\tilde{d}(G)=n/2$.
See Figure \ref{fig:PoF}.

In a maximum arrangement on $G$, every agent in $\agents_K$ is assigned to \revised{the clique of size $n$} and every agent in $\agents_C$ is assigned to disjoint edges.
The social welfare is $n(n-1)+n/2=n(n-1/2)$ and at least one agent on an edge has the utility 0.

On the other hand, in a maximin arrangement, every agent in $\agents_C$ is assigned to a clique and every agent in $\agents_K$ is assigned to disjoint edges.
Then the utility of any agent is 1 and the social welfare is $n+n=2n$.
Therefore, the price of fairness is \revised{$n(n-1/2)/2n=n/2-1/4= \tilde{d}(G)-1/4$}. 

For the symmetric case, we modify the preference graph and the seat graph.
For the preference graph, we change a directed cycle in the preference graph to an undirected cycle. For the seat graph, we replace $n/2$ disjoint edges by $n/3$ disjoint triangles $K_3$.
When every agent in $\agents_K$ is assigned to a clique, the social welfare is $n(n-1)+4n/3=n(n+1/3)$ and it is maximum. The least utility of an agent is 1.
On the other hand, when every agent in $\agents_C$ is assigned to a clique, the utility of any agent is 2.
The social welfare is $2n+6n/3=4n$.
Therefore, the price of fairness is \revised{$n(n+1/3)/4n=n/4+1/12= \tilde{d}(G)/2+1/12$}. 
\end{proof}

\section{Computational Complexity}\label{sec:complexity}
In this section, we give the dichotomy of computational complexity of three \textsc{Seat Arrangement} problems in terms of the order of components in the seat graph.

\subsection{Tractable case}
In this subsection, we show that  {\sf UTA}, {\sf EGA}, and symmetric {\sf EFA}  are solvable in polynomial time if each component of the seat graph is of order at most 2. 
\begin{theorem}\label{MU:forest:P}
 {\sf UTA} is solvable in polynomial time if each connected component of the seat graph has at most two vertices.
\end{theorem}
\begin{proof} 
Let $K_n=(\agents,E_{\agents})$ be the weighted and undirected complete graph such that the weight of edge $\{p,q\}\in E_{\agents}$ is defined by $f_p(q)+f_q(p)$.
Also, let  $n'$ be the number of endpoints of edges in the seat graph. Notice that $n'$ is always even.
\revised{Then, we find a maximum weight matching $M_{n'/2}$ of size $n'/2$ in $K_n$ in polynomial time by using Theorem~\ref{thm:k-matching}.}
Next, we assign each pair of the agents in $M_{n'/2}$ to an edge in $G$ and the rest of the agents \rerevised{to} isolated vertices.
Let $\pi^*$ be such an arrangement.

In the following, we show that $\pi^*$ is maximum.
Suppose that  there exists $\pi'$ such that $\sw(\pi')>\sw(\pi^*)$.
Let $V_E$ be the set of endpoints of $E$ and $V_I$ be the set of isolated vertices.
Since the size of each connected component in $G$ is  bounded by 2, for any $\pi$, we have $\sw(\pi)=\sum_{v\in V_E} U_{\pi^{-1}(v)}(\pi) + \sum_{v\in V_I} U_{\pi^{-1}(v)}(\pi)
=\sum_{v\in V_E} U_{\pi^{-1}(v)}(\pi) 
=\sum_{(u,v)\in E} (f_{\pi^{-1}(u)}({\pi^{-1}(v)}) 
+f_{\pi^{-1}(v)}({\pi^{-1}(u)}))$.
Now, the number of endpoints of edges \rerevised{in $E$} is $n'$ \rerevised{and} $|E|=n'/2$.
\rerevised{Then,} for any $\pi$, $M=\{\{\pi^{-1}(u), \pi^{-1}(v)\}\mid \{u,v\}\in E\}$ is a matching of size $n'/2$ with weight $\sw(\pi)=\sum_{(u,v)\in E} (f_{\pi^{-1}(u)}({\pi^{-1}(v)})+f_{\pi^{-1}(v)}({\pi^{-1}(u)})=\sum_{(p,q)\in M} (f_{p}(q)+f_{q}(p))$ in $K_n$.
Thus, if there exists $\pi'$ such that $\sw(\pi')>\sw(\pi^*)$, there exists a heavier matching of size $n'/2$ than $M_{n'/2}$.
This contradicts that $M_{n'/2}$ is a maximum weight matching of size $n'/2$.
\end{proof}

By using a maximin matching algorithm proposed in \cite{Gabow1988} instead of Edmonds's algorithm,  we can solve  {\sf EGA} in polynomial time.
We apply that to the weighted and undirected complete graph such that the weight of edge $\{p,q\}\in E_{\agents}$ is defined by $\min \{ f_p(q), f_q(p)\}$.
\begin{theorem}\label{EGA:forest:P}
 {\sf EGA}  is solvable in polynomial time if each connected component of the seat graph has at most two vertices.
\end{theorem}
\begin{proof}
Let $K_n=(\agents,E_{\agents})$ be the weighted and undirected complete graph such that the weight of edge $\{p,q\}\in E_{\agents}$ is defined by $\min\{f_p(q), f_q(p)\}$.
Also, let  $n'$ be the number of endpoints of edges in the seat graph.
Next, we find a maximin matching $M$ of size $n'$ in $K_n$, which is a matching of size $n'$ such that the minimum weight of edges in $M$  is maximum.
It can be computed in time $O(m\sqrt{n\log n})$~\cite{Gabow1988}.
Then we assign  $n'/2$ pairs of agents in $M$ to endpoints of an edge in $G$ and the rest of the agents to isolated vertices.

In the following, we confirm that such an arrangement, denoted by $\pi$, is a maximin arrangement. 
If the least utility among all the agents on $\pi$ is 0 and there is at least one isolated vertex, it is clearly a maximin arrangement because an agent with the least utility is on an isolated vertex.
\revised{Otherwise, an agent having the least utility among all the agents is assigned to an endpoint of an edge.
Here, we denote the agent having the least utility in $\pi$ by $p_l^{\pi}$ and its utility by $U_{p_l^{\pi}}(\pi)$.}

Suppose that $\pi$ is not a maximin arrangement.
Then there exists an assignment $\pi'$ and an agent $p_l^{\pi'}$ with the least utility in $\pi'$ such that $U_{p_l^{\pi'}}(\pi')> U_{p_l^{\pi}}(\pi)$.
\revised{We also denote an agent $p^{\pi'}_{l_e}$ with the least utility among the agents assigned to an endpoint of an edge in $\pi'$.}
Note that $U_{p_{l_e}^{\pi'}}(\pi')\ge U_{p_l^{\pi'}}(\pi')$ and \revised{it is possible that $p_l^{\pi'}=p^{\pi'}_{l_e}$}. 
Let $q'$ be the partner of $p_{l_e}^{\pi'}$ on the edge. 
Since the seat graph consists of $n'/2$ edges, we have $U_{p_{l_e}^{\pi'}}(\pi')=f_{p_{l_e}^{\pi'}}(q')$.
Here, we define $M'=\{\{\pi'^{-1}(u), \pi'^{-1}(v)\}\mid \{u,v\}\in E\}$. Then $M'$ is a matching of size $n'/2$ in $K_n$.
Because the weight of an edge in $K_n$ is defined as $\min\{f_p(q), f_q(p)\}$, \revised{the least weight among edges in $M'$ is $U_{p_{l_e}^{\pi'}}(\pi')=f_{p_{l_e}^{\pi'}}(q')$.}
Since we have $U_{p_{l_e}^{\pi'}}(\pi')> U_{p_l^{\pi}}(\pi)$, this contradicts that $M$ is a maximin matching of size $n'/2$.
\end{proof}

For {\sf EFA}, we show several cases that can be solved in polynomial time. 
The following theorem can be shown by taking a perfect matching on the {\em best-preference graph} $G^{\rm best}_{\preferencefamily_{\agents}}=(\agents, E'_{\preferencefamily_{\agents}})$, where $E'_{\preferencefamily_{\agents}} = \{\{p,q\} \in E_{\preferencefamily_{\agents}} \mid f_p(q) \geq f_p(q')  \text{ for all } q'\in \agents \setminus \{p\} \text{ and } f_q(p) \geq f_q(p') \text{ for all }  p'\in \agents \setminus \{q\} \}$. Note that  $G^{\rm best}_{\preferencefamily_{\agents}}$ is a bidirectional graph and hence it can be regarded as an undirected graph.

\begin{theorem}\label{thm:envy:matching}
{\sf EFA} can be solved in polynomial time if each connected component of the seat graph is an edge.
\end{theorem}
 \begin{proof}
\revised{We first observe that each agent must match to the most preferable agent on an edge.
If not so, an agent that does not match to the most preferable agent envies the agent that matches to it.
Thus, we consider the {best preference graph} $G^{\rm best}_{\preferencefamily_{\agents}}=(\agents, E'_{\preferencefamily_{\agents}})$, where $E'_{\preferencefamily_{\agents}} = \{\{p,q\} \in E_{\preferencefamily_{\agents}} \mid f_p(q) \geq f_p(q')  \text{ for all } q'\in \agents \setminus \{p\} \text{ and } f_q(p) \geq f_q(p') \text{ for all }  p'\in \agents \setminus \{q\} \}$.}

\revised{As we observed above, $p$ and $q$ are matched in an envy-free arrangement in $G$ if and only if $f_p(q) \geq f_p(q')$  for all $q'\in \agents \setminus \{p\}$ and $f_q(p) \geq f_q(p')$for all   $p'\in \agents \setminus \{q\} \}$, which means that $p$ and $q$ are matched in an envy-free arrangement if and only if there exists an edge $\{p,q\}\in E'_{\preferencefamily_{\agents}}$.}

\revised{Therefore, we can see that there is a perfect matching in $G^{\rm best}_{\preferencefamily_{\agents}}$ if and only if there is an envy-free arrangement in $G$. Such a perfect matching can be computed in polynomial time and we complete the proof.}
\end{proof}

\begin{theorem}\label{thm:symmetric:envy}
Symmetric \textsf{EFA} can be solved in polynomial time if each connected component of the seat graph has at most two vertices.
\end{theorem}
\begin{proof}
Let $(G, \agents, \preferencefamily_{\agents})$ be an instance of symmetric \textsf{EFA}.
We denote by $E$ and $I$ the sets of edges and isolated vertices of $G$, respectively.
By Theorem~\ref{thm:envy:matching}, we may assume that $I \ne \emptyset$.
Now, if there is an agent $p$ such that $f_p(q)<0$ for all $q\in \agents \setminus \{p\}$,
it has to be assigned to an isolated vertex $v \in I$.
If not so, $p$ envies the agents assigned to isolated vertices.
Thus, we can reduce the instance $(G, \agents, \preferencefamily_{\agents})$ to  $(G-v, \agents\setminus \{p\}, \preferencefamily_{\agents\setminus \{p\},})$.
In the following, we assume that there is no such agent.

Let $\agents_{0} = \{p \in \agents \mid \max_{q \in \agents \setminus \{p\}} f_{p}(q) = 0\}$
and $\agents_{+} = \{p \in \agents \mid \max_{q \in \agents \setminus \{p\}} f_{p}(q) > 0\}$.
The assumption above implies that $\agents = \agents_{0} \cup \agents_{+}$.
Let $H_{0} = (\agents_{0}, E_{0})$ be the undirected graph with $E_{0} = \{\{p,q\} \mid f_{p}(q) = f_{q}(p) = 0\}$.
Similarly, let $H_{+} = (\agents_{+}, E_{+})$ be the undirected graph with $E_{+} = \{\{p,q\} \mid f_{p}(q) = f_{q}(p) > 0\}$.

Let $C$ be a connected component of $H_{+}$.
Observe that if $C$ contains an edge $\{u,v\}$ such that 
$u$ is assigned to an endpoint of $e \in E$
and $v$ is assigned to a vertex in $I$,
then $v$ envies the agent assigned to the other endpoint of $e$.
Since $C$ is connected, this implies that in every envy-free arrangement
either all agents in $C$ are assigned to $E$
or all agents in $C$ are assigned to $I$.

Observe also that if two agents with a negative mutual preference are assigned to an edge in $E$,
then they envy the agents assigned to $I$.
Therefore, for each $e \in E$, the agents assigned to $e$ are adjacent in $H_{0}$ or $H_{+}$.
Furthermore, if an agent $p \in \agents_{+}$ is assigned to an endpoint of $e \in E$,
then all neighbors of $p$ in $H_{+}$ are assigned to endpoints of some edges in $E$, and thus
the agent assigned to the other endpoint of $e$ has to be one of the most preferable ones
to make $p$ envy-free.

The discussion so far implies that 
there is an envy-free arrangement if and only
$E$ can be completely packed with
some connected components of $H_{+}$ and some edges of $H_{0}$
so that the agents assigned to each $e \in E$ are best-preference pairs.

For each component $C$ of $H_{+}$,
we check whether $C$ has a perfect matching that uses only the best-preference edges.
If $C$ has no such matching, then it has to be packed into $I$.
Let $C_{1}, \dots, C_{h}$ be the connected components of $H_{+}$ with such perfect matchings.

We now compute the maximum value $k$ of the \textsc{Maximum 0-1 Knapsack} instance with $h$ items such that the weight and the value of the $i$th item are $|C_i|/2$ and the budget is $m = |E|$.
This can be done in time polynomial in $n$~\cite{GJ1979}.
We can see that there is an envy-free arrangement if and only if $k + k' \ge m$,
where $k'$ is the size of a maximum matching of $H_{0}$.
An envy-free arrangement can be constructed by 
first packing the best-preference perfect matchings in the components corresponding to the chosen items into $E$,
and then packing a matching of size $m - k \le k'$ in $H_{0}$ into the unused part of $E$.
\end{proof}


\begin{theorem}
{\sf EFA} can be solved in polynomial time if  each connected component has at most two vertices and the preferences are positive.
\end{theorem}
\begin{proof}
If the preferences are positive, whenever there is an isolated vertex in the seat graph, then an agent assigned to it envies other agents assigned to an edge. Thus, we can suppose that there is no isolated vertex and apply Theorem \ref{thm:envy:matching}.
\end{proof}

\rerevised{Recently,} Ceylan, Chen, and Roy \rerevised{have shown} that {\sf EFA} 
remains NP-hard even for seat graphs where each connected \revised{component} has at most two vertices and the preferences are strict~\cite{SA:Ceylan0R23}. Thus, the boundary of the complexity of {\sf EFA} is shown.


\subsection{Intractable case}
First, we show that \textsc{Stable Roommates}  with the complete preference lists under exchange stability  can be transformed into 
{\sf STA}.
Let $n$ be the number of agents in \textsc{Stable Roommates}.
According to the complete preference order of agent $p$ in \textsc{Exchange Stable Roommates}, one can assign values from $1$ to $n-1$ to the preferences of $p$ to other agents in {\sf STA}. Moreover, let $G$ be the seat graph consisting of $n/2$ disjoint edges.
Then it is easily seen that there is a stable matching if and only if there is a stable arrangement in $G$. Since \textsc{Stable Roommates} with complete \revised{and strict} preference lists under exchange stability is NP-complete \cite{Cechlarova2002,Cechlarova2005}, {\sf STA} is also NP-complete.
\begin{theorem}
{\sf STA} is NP-complete even if the preferences are positive \revised{and strict,} and each component of the seat graph is of order two.
\end{theorem}

We also show that {\sf StrictSTA} is NP-complete.
\begin{theorem}
{\sf StrictSTA} is NP-complete even if the preferences are positive and strict, and each component of the seat graph is of order two.
\end{theorem}
\begin{proof}
    It is sufficient to show that there exists a stable assignment on the seat graph $G$ of order two if and only if there exists a strictly stable assignment on $G$ under positive and strict preferences.

    By definition, the opposed direction is trivial. For the forward direction, suppose that a stable assignment $\pi$ is not strictly stable. Then there exists a blocking pair $(p,q)$ such that the utility of $p$ strictly increases and the utility of $q$ does not decrease after the swap. Since the preferences are strict, after the swap, the utility of $q$ also increases, which contradicts the fact that $\pi$ is stable.
\end{proof}

\color{black}

Then we prove that symmetric {\sf EFA} is NP-complete even if each connected component has at most three vertices.
\begin{theorem}\label{thm:EFA:symmetric:hard}
Symmetric {\sf EFA} is NP-complete even if each connected component of the seat graph has at most three vertices.
\end{theorem}
\begin{proof}
We give a reduction from \textsc{Partition into Triangles}: given a graph $G=(V,E)$, determine whether $V$ can be partitioned into 3-element sets $S_1, \ldots, S_{|V|/3}$ such that each $S_i$ forms a triangle $K_3$ in $G$.
The problem is NP-complete  \cite{GJ1979}.

Given a graph $G=(V,E)$, we construct the instance of {\sf EFA}.
First, we set $\agents=V\cup \{x,y,z\}$. Three agents $x,y,z$ are called {\em super agents}.
Then we define the preferences as follows.
For $p,q\in \{x,y,z\}$, we set $f_p(q)=f_q(p)=2$.
For $p\in \{x,y,z\}, q\in V$, we set $f_p(q)=f_q(p)=1$.
Finally, for $p,q\in V$, we set $f_p(q)=f_q(p)=1$ if $\{p,q\}\in E$, and otherwise, 
$f_p(q)=f_q(p)=0$.
Clearly, the preferences are symmetric.
The seat graph $H$ consists of $|V|/3+1$ disjoint triangles.

Given a partition $S_1, \ldots, S_{|V|/3}$ of $V$, we assign them to triangles in the seat graph. Moreover, we assign $\{x,y,z\}$ to a triangle.
Then the utilities of agents in $V$ are 2 and the utilities of $x,y,z$ are 4, respectively.
Since these utilities are maximum for all agents, this arrangement is envy-free.

Conversely, we are given an envy-free arrangement $\pi$.
\begin{nestedclaim}\label{claim:envy:hardness}
In any envy-free arrangement $\pi$, $\{x,y,z\}$ is assigned to the same triangle in $H$.
\end{nestedclaim}
\begin{claimproof}
Suppose that $x$ is assigned to a triangle $T_x$ and $y$ is assigned to another triangle $T_y$.
If $z$ is assigned to $T_x$, $y$ envies the agent $p$ in $V$ assigned to $T_x$ because the utility of $y$ is increased from 2 to 4 by swapping $y$ and $p$.
Similarly, if  $z$ is assigned to another triangle $T_z$, $y$ envies the agent in $V$ assigned to $T_x$ because the utility of $y$ is increased from 2 to 3 by swapping $y$ and $p$.
Thus, $\{x,y,z\}$ must be assigned to the same triangle.
\end{claimproof}
Then, if three agents $p,q,w\in V$ such that $\{p,q\}\notin E$ are assigned to the same triangle in $H$, $p$ envies $x$  because the utility of $p$ is increased from 1 to 2 by swapping $p$ and $x$.
Since $\pi$ is envy-free, for each triangle assigned to $p,q,r\in V$, they satisfy $\{p,q\},\{q,r\},\{r,p\}\in E$.
This implies that there is a partition $S_1, \ldots, S_{|V|/3}$ of $V$ in $G$.
\end{proof}

By definition, if the seat graph is a complete graph, all the arrangements are equivalent.
However, if the seat graph consists of a clique and an independent set, {\sf EFA} is NP-complete. The reduction is from \textsc{$k$-Clique}, \rerevised{which is NP-complete \cite{GJ1979}.}
\begin{theorem}\label{thm:EFA:clique}
{\sf EFA} is NP-complete even if the seat graph consists of  a clique and an independent set.
\end{theorem}
\begin{proof}

We give a reduction from \rerevised{\textsc{$k$-Clique}}.
We are given a graph $G=(V,E)$. 
For each $e\in E$, we use the corresponding agent $p_e$.
 For each $v\in V$, we make $M$ agents $p_v^{(1)}, \ldots, p_v^{(M)}$, \revised{where $M=n^3$}.
The number of agents is $|E|+M|V|$.
Then we define the preferences of agents.
For $p_e$, we define $f_{p_e}(p_u^{(1)})=f_{p_e}(p_v^{(1)})=1$ if \rerevised{$e=\{u,v\}$}, and otherwise $f_{p_e}(p_u^{(1)})=f_{p_e}(p_v^{(1)})=0$. 
For each $v\in V$ and $i, j$, we set $f_{p_v^{(i)}}(p_v^{(j)})=f_{p_v^{(j)}}(p_v^{(i)})=1$. \revised{For the remaining preferences for $p$ and $q$ still not defined, we set $f_p(q)=f_q(p)=0$.}
Finally, we define the seat graph $G'=(I\cup C, E')$ as a graph consisting of an independent set $I$ of size $Mk+k(k-1)/2$ and a clique $C$ of size $|E|+(|V|-k)M-k(k-1)/2$.

In the following, we show that there is a clique of size $k$ in $G$ if and only if there is an envy-free arrangement in $G'$.
Given a clique of size $k$, we assign all agents $p^{(i)}_v$ and $p_e$ corresponding to a clique to vertices in $I$. 
Since the number of such agents is $Mk+k(k-1)/2$, the set of vertices not having \revised{an} agent in $G'$ is $C$.
Thus, we assign other agents to vertices in $C$.
Because the utility of an agent on $I$ does not increase even if he is swapped for any agent on $C$, every agent on $I$ is envy-free.
Moreover, an agent on $C$ is envy-free because the utility is at least 1 and $C$ is a clique.
Therefore, such an arrangement is envy-free.

Conversely, we are given an envy-free arrangement $\pi$.
First, we observe the following fact.
\begin{fact}\label{fact1}
 If some $p^{(i)}_v$ is on $I$, all the $p^{(j)}_v$ for $j\neq i$ must be on $I$.
 \end{fact}
Otherwise, $p^{(i)}$ envies some agent on $C$ because the utility increases by moving to $C$.
Also, we have the following fact.
\begin{fact}\label{fact2}
If $p_e$ where $e=\{u,v\}$ is on $I$, all the $p^{(i)}_u$ and $p^{(i)}_v$ must be on $I$.
 \end{fact}
If not so, $p^{(1)}_u$ is on $C$ because every $p^{(i)}_u$ must be on $C$ by Fact 1. 
However, this implies that  $p_e$ envies some agent on $C$.
This is a contradiction.

Now, since $|I| = Mk+k(k-1)/2$ \revised{and $M=n^3$}, at most $Mk$ $p^{(i)}_v$'s are on $I$ from Fact 1.
In other words, there are at most $k$ vertices in $V$ such that $p^{(i)}_v$ is on $I$ for all $i$.
The remaining vertices in $I$  have $p_e$.
From Facts 1 and 2, for every $e=\{u,v\}\in E$ such that $p_e\in I$, all $p^{(i)}_u$ and $p^{(i)}_v$ for all $i$ must be on $I$.
Because the number of $p_e$'s on $I$ is at least $k(k-1)/2$, $I$ must have exactly $k(k-1)/2$ $p_e$'s and $Mk$  $p^{(i)}_v$'s such that $v$ is an endpoint of $e$,  so that $\pi$ is envy-free.
This implies that $\{v\mid p^{(1)}_v\in I\}$ is a clique of size $k$ in $G$.
\end{proof}

Finally, we show that {\sf EFA} is NP-complete even if both the preference graph and the seat graph are restricted. The reduction is from \textsc{3-Partition} \cite{GJ1979}.
\begin{theorem}\label{thm:envy-free:preDAG}
{\sf EFA} is NP-complete even if the preference graph is a directed acyclic graph (DAG) and the seat graph is a tree.
\end{theorem}
\begin{proof}
We give a reduction from \textsc{3-Partition}.
\revised{Given a set of integers $A=\{a_1, \ldots, a_{3n}\}$, the problem is to find a partition $(A_1, \ldots, A_n)$ such that $|A_i|=3$ and $\sum_{a\in A_i}a=B$ for each $i$ where $B=\sum_{a\in A} a/n$.} We call such a partition a {\em 3-partition}. \revised{\textsc{3-Partition} remains strongly NP-complete even if every integer in $A$ is strictly between $B/4$ and $B/2$~ \cite{GJ1979} .}

First, we prepare $n$ agents $P_T=\{p_{t_1}, \ldots, p_{t_n}\}$ corresponding to the resulting triples and $3n$ agents  $P_A=\{p_{a_1}, \ldots, p_{a_{3n}}\}$ corresponding to elements.
Moreover, we use an agent $p_r$, called a {\em root} agent.
Then we define the preferences and the seat graph as in Figure \ref{fig:fig_EFA_DAG}.
\revised{For each $p_{t}\in P_T$, we set $f_{p_{t}}(p_r)=nB$ and $f_{p_{t}}(p_a)=a$}.
\revised{The remaining preference from $p\in \agents$ to $q\in \agents$ is set to $0$.}
Note that the preference graph is a DAG \revised{(see Figure \ref{fig:fig_EFA_DAG} (left)).}
Then we define the seat graph $G=(V,E)$, which is a tree with the root vertex $v_r$.
The root vertex has $n$ children and \revised{each of its children has} exactly three children.
The number of vertices is $4n+1$.

\begin{figure}[tbp]
 \begin{center}
  \includegraphics[width= \textwidth]{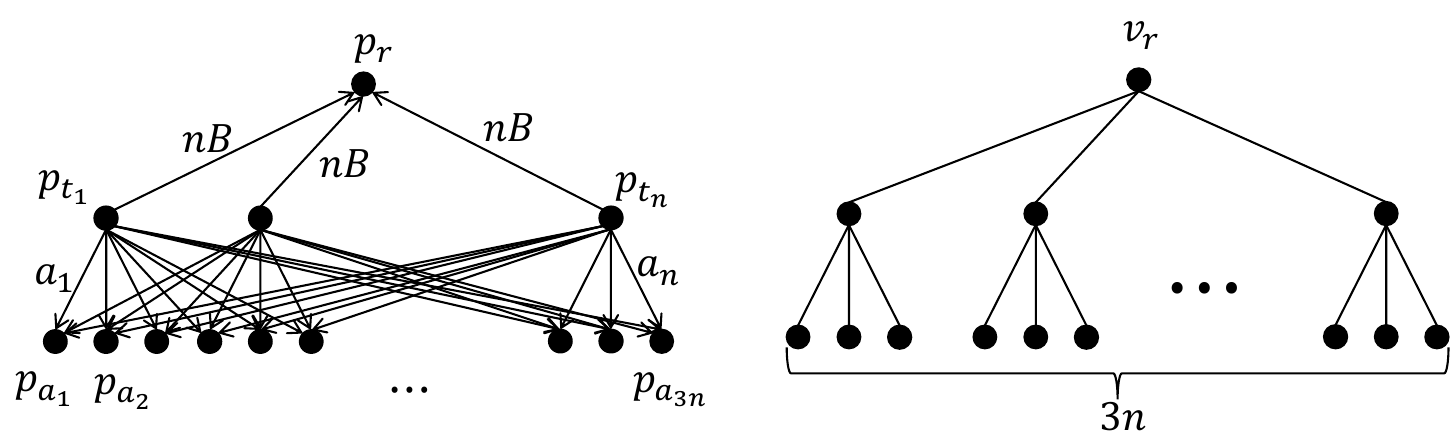}
 \end{center}
 \caption{The preference graph (left) and the seat graph in the proof of Theorem \ref{thm:envy-free:preDAG}}
 \label{fig:fig_EFA_DAG}
\end{figure}

Given a 3-partition $(A_1, \ldots, A_n)$, we assign root agent $p_r$ to $v_r$.
Moreover, for three elements in $A_i$, we assign the three corresponding agents to leaves with the same parent in $G$.
Finally, we assign $p_t$ to an inner vertex in $G$ arbitrarily.

We show that this arrangement, denoted by $\pi$, is envy-free.
By the definitions of preferences, the utilities of $p_r$ and $p_a\in P_A$ are $0$ for any arrangement.
Thus, they are envy-free.
Each $p_t\in P_T$ is also envy-free because every utility of $p_t$ is \revised{$(n+1)B$} and the preferences of $p_t$ to agents in $P_A\cup \{p_r\}$ are identical.
 
Conversely, we are given an envy-free arrangement $\pi$.
Suppose that $p_r$ is not assigned to the root vertex $v_r$.
\revised{Then the degree of $\pi(r)$ is at most 4, and there exists an agent  $q_t\in P_T$ not adjacent to $p_r$ on $\pi$. Since $a_i<B/2$ and the maximum degree of the seat graph is $n$, the utility of $q_t$ is at most $nB/2$. Due to $f_{q_t}(p_r)=nB$, $q_t$ envies a neighbor of $p_r$.
Therefore, $p_r$ must be assigned to $v_r$.}

If an agent $p_t$ in $P_T$ is assigned to a leaf in $G$, $p_t$ does not have $p_r$ as a neighbor.
Thus, the utility of $p_t$ is less than $B/2$, and $p_t$ envies a neighbor of $p_r$. \revised{Therefore, every $p_t$ must be assigned to inner vertices in $T$.
Since the number of inner vertices is $n$, every agent in $P_A$ must be assigned to a leaf.
If there is $p_t\in P_T$ with utility more than $(n+1)B$, there is $p'_t\in P_T$  with utility less than $(n+1)B$ since $\sum_{a\in A} a=nB$ and $f_{p_{t}}(p_r)=nB$ for every $p_{t}\in P_T$.}
In this case, $p'_t$ envies $p_t$.
Therefore, the utility of each agent in $P_T$ is exactly \revised{$(n+1)B$}.
Because $f_{p_{t}}(p_r)=B$ for every $p_{t}$,  if we partition $A$ according to neighbors of $p_t$, the resulting partition is a 3-partition.
\end{proof}

Next, we show that {\sf UTA} and {\sf EGA} are NP-complete for several graph classes by reductions from \textsc{Spanning Subgraph Isomorphism}.
Here, we define \textsc{Spanning Subgraph Isomorphism} as follows:
given two graphs $G=(V(G),E(G))$ and $H=(V(H),E(H))$ where $|V(G)|=|V(H)|$, determine whether there is a bijection $g: V(G)\to V(H)$ such that $\{g(u), g(v)\}\in E(H)$ for any $\{u,v\}\in E(G)$. 

\textsc{Spanning Subgraph Isomorphism} is NP-complete even if $G$ is a path \revised{or} a cycle by a reduction from \textsc{Hamiltonian Path} and \textsc{Hamiltonian Cycle}~\cite{GJ1979}.
\revised{Moreover, it is NP-complete if $G$ is a proper interval graph, a trivially perfect graph, a split graph, or a bipartite permutation graph~\cite{Kijima2012}.}
When $G$ is disconnected, it is also NP-complete even if $G$ is a forest and a cluster graph whose components are of size three~\cite{GJ1979,Bod2019}. Here, if $G$ is in some graph class ${\mathcal G}$ in \textsc{Spanning Subgraph Isomorphism}, we call the problem \textsc{Spanning Subgraph Isomorphism} of ${\mathcal G}$.
Then we give the following theorems.
\begin{theorem}\label{spanning:NP}
If \textsc{Spanning Subgraph Isomorphism} of ${\mathcal G}$ is NP-complete, then  {\sf UTA} on ${\mathcal G}$ is NP-hard even if the preferences are binary and symmetric.
\end{theorem}
\begin{proof}
Given an instance of \textsc{Spanning Subgraph Isomorphism} $(G,H)$, we construct an instance of {\sf UTA} as follows.
Let $\agents=V(H)$ be the set of agents and $G$  be the seat graph.
Then we set the preferences of agents as follows: 
\begin{align*}
\begin{cases}
f_p(q)=f_q(p)=1& \mbox{if $\{p,q\}\in E(H)$}\\
f_p(q)=f_q(p)=0 &\mbox{otherwise.}
\end{cases}
\end{align*}
That is, the preference graph is $H$.
By definition, the preferences of agents are symmetric.
We complete the proof by showing that there is a bijection $g$ such that  $\{g(u), g(v)\}\in E(H)$  for any $\{u,v\}\in E(G)$ if and only if there exists an arrangement $\pi$ with social welfare $2|E(G)|$ in $G$.

Let  $\pi=g^{-1}$.
Since bijection $g$ satisfies that $\{g(u), g(v)\}\in E(H)$  for any $\{u,v\}\in E(G)$  and $f_{g(u)}(g(v))=f_{g(v)}(g(u))=1$ for $\{g(u),g(v)\}\in E(H)$,  there exists $\{u,v\}\in E(G)$ for any $p,q\in \agents (=V(H))$ such that $\pi(p)=u, \pi(q)=v$, and $f_{p}(q)=f_{q}(p)=1$.
Thus, $U_p(\pi)=\sum_{v\in N(\pi(p))} f_p(\pi^{-1}(v))=d_G(\pi(p))$ for any $p\in \agents(=V(H))$.
Finally, we have the social welfare $\sw(\pi)=\sum_{p\in \agents}U_p(\pi)=\sum_{p\in \agents}d_G(\pi(p))=\sum_{v\in V(G)}d_G(v)=2|E(G)|$.

Conversely, we are given $\pi$ with social welfare $2|E(G)|$ in $G$.
Let $g=\pi^{-1}$.
Suppose that there is an edge $\{u,v\}\in E(G)$ such that $\{\pi^{-1}(u),\pi^{-1}(v)\}\notin E(H)$.
Then, it holds that $f_{\pi^{-1}(u)}(\pi^{-1}(v))=f_{\pi^{-1}(v)}(\pi^{-1}(u))=0$ by the definition of the preferences.
Thus, there exists an agent $p=\pi^{-1}(u)\in \agents$ such that $ U_p(\pi)<d_G(\pi(p))$ since it holds that $U_p(\pi)\le d_G(\pi(p))$ for any $p\in \agents$.
This implies that ${\sf sw}(G)<\sum_{p\in \agents}d_G(\pi(p))=2|E(G)|$.
This is a contradiction.
Thus, there exists a bijection $g=\pi^{-1}$ such that  $\{\pi^{-1}(u), \pi^{-1}(v)\}\in E(H)$  for any $\{u,v\}\in E(G)$.
This completes the proof.
\end{proof}

By assuming that the seat graph $G$ is regular, we obtain the following theorem.
\begin{theorem}\label{spanning:NP:maximin}
If \textsc{Spanning Subgraph Isomorphism} of a class $\mathcal{G}$ of regular graphs is NP-complete, then {\sf EGA} on $\mathcal{G}$ is NP-hard even if the preferences are binary and symmetric.
\end{theorem}
 \begin{proof}
We give a reduction from  \textsc{Spanning Subgraph Isomorphism} to {\sf EGA}.
The setting is the same as {\sf UTA}.
Let $G$ be an $r$-regular graph.
Then we show that there is a bijection $g$ such that  $\{g(u), g(v)\}\in E(H)$  for any $\{u,v\}\in E(G)$ if and only if there exists an arrangement $\pi$ such that the least
utility of an agent is $r$ in $G$.
Given a bijection $g$ such that $\{g(u), g(v)\}\in E(H)$ for any $\{u,v\}\in E(G)$, we set $\pi=g^{-1}$.
 Since bijection $g$ satisfies that $\{g(u), g(v)\}\in E(H)$  for any $\{u,v\}\in E(G)$  and $f_{g(u)}(g(v))=f_{g(v)}(g(u))=1$ for $\{g(u),g(v)\}\in E(H)$, there exists \rerevised{$\{u,v\}$} in $E(G)$ for any $p,q\in \agents(=V(H))$ such that $\pi(p)=u, \pi(q)=v$, and $f_{p}(q)=f_{q}(p)=1$.
Thus, $U_p(\pi)=\sum_{v\in N(\pi(p))} f_p(\pi^{-1}(v))=r$ for any $p\in \agents(=V(H))$ since $G$ is $r$-regular.

Conversely, we are given an arrangement $\pi$ such that the least
utility of an agent is $r$ in $G$.
Let $g=\pi^{-1}$.
Suppose that there is an edge $\{u,v\}\in E(G)$ such that $\{\pi^{-1}(u),\pi^{-1}(v)\}\notin E(H)$.
Then, it holds that $f_{\pi^{-1}(u)}(\pi^{-1}(v))=f_{\pi^{-1}(v)}(\pi^{-1}(u))=0$ by the definition of the preferences.
Thus, there exists an agent $p=\pi^{-1}(u)\in \agents$ such that $ U_p(\pi)<r$ since it holds that $U_p(\pi)\le d_G(\pi(p))$ for any $p\in \agents$.
This is a contradiction.
Thus, there exists a bijection $g=\pi^{-1}$ such that  $\{\pi^{-1}(u), \pi^{-1}(v)\}\in E(H)$  for any $\{u,v\}\in E(G)$.
\end{proof}

\begin{corollary}\label{Thm:maxutility:NP-c:path}
{\sf UTA} and {\sf EGA} are  NP-hard on cycles, and cluster graphs whose components are of order three.
Furthermore, {\sf UTA} is NP-hard on paths and linear forests whose components are paths of length three.
These hold even if the preferences are binary and symmetric.
\end{corollary}

Moreover,  \textsc{Spanning Subgraph Isomorphism} cannot be solved in time $n^{o(n)}$ unless \ETH{} is false \cite{Cygan2017} where $n=|V(G)|=|V(H)|$.
Since the reduction in Theorem \ref{spanning:NP} satisfies $|\agents|=n$, we also have the following result.
\begin{corollary}\label{Thm:maxutility:ETH}
{\sf UTA} cannot be solved in time $n^{o(n)}$ unless \ETH{} is false.
\end{corollary}

\section{Parameterized Complexity}\label{sec:parameterized_complexity}
{\sf UTA} is NP-hard even on trees  (i.e., treewidth 1), which implies that it admits
no parameterized algorithm by treewidth if $\mathrm{P} \neq \mathrm{NP}$. Thus, we consider designing an algorithm parameterized by a \rerevised{more restricted} parameter: vertex cover number.

\begin{theorem}\label{parameterized_vc_mca}
{\sf UTA} can be solved in time $O(n^\vc \vc!(n-\vc)^{3})$ where $\vc$ is the vertex cover number of the seat graph.
\end{theorem}
\begin{proof}
Given an instance $(G, \agents, \preferencefamily_\agents)$ of {\sf UTA}, we first compute a minimum vertex cover $S$ of size $\vc$ in time $O(1.2738^\vc + \vc n)$~\cite{Chen2010}.
Then we guess $\vc$ agents that are assigned to vertices in $S$.
Let $\agents'$ be the set of agents assigned to $S$.
Next, we guess all arrangements that assign $\agents'$ to $S$.
The number of candidates of arrangements is $O(n^\vc \vc!)$.

For each candidate, we consider how to assign the rest of the agents in $\agents \setminus \agents'$ to $V\setminus S$.
Since $V\setminus S$ is an independent set, we can compute the utility of an agent $p\in \agents \setminus \agents'$ when $p$ is assigned to $v\in V\setminus S$.
Note that $p$ does not affect the utility of other agents in $\agents \setminus \agents'$.
Moreover, we can also compute the increase of the utilities of neighbors of $p$  when $p$ is assigned to $v\in V\setminus S$.
Then we observe that the increase of the social welfare when $p$ is assigned to $v\in V\setminus S$ is the utility of \rerevised{the} agent $p$ \rerevised{plus the sum of the preferences from neighbors of $p$ to $p$, that is, $\sum_{u\in N_G(v)} {f_p(\pi^{-1}(u))} + \sum_{u\in N_G(v)}{f_{\pi^{-1}(u)}(p)}$.}
Thus, by computing a maximum weight perfect matching on a complete bipartite graph $(\agents\setminus \agents', V\setminus S; E')$ with edge weight $w_{pv}=\sum_{u\in N_G(v)} ({f_p(\pi^{-1}(u))} + {f_{\pi^{-1}(u)}(p)})$ for every candidate, we can obtain a maximum arrangement in $G$.

Since we can compute a maximum weight perfect matching in time $O((n-\vc)^{3})$~\cite{Gabow17}, the total running time is $O(n^\vc \vc!(n-\vc)^{3})$.
\end{proof}
By Proposition \ref{prop:stability:Massand}, we obtain the following corollary.
\begin{corollary}\label{parameterized_vc_symmetric_sta}
Symmetric {\sf STA} can be solved in time $O(n^\vc \vc!(n-\vc)^{3})$ where $\vc$ is the vertex cover number of the seat graph.
\end{corollary}
Then we give \rerevised{a} tight lower bound for {\sf UTA} parameterized by vertex cover number.
\begin{theorem}\label{thm:Whard:UTA}
{\sf UTA} is W[1]-hard parameterized by the vertex cover number ${\vc}$ of the seat graph even if the preferences are binary.
Furthermore, there is no $f({\vc})n^{o({\vc})}$-time algorithm unless ETH fails where $f$ is some computable function.
\end{theorem}
\begin{proof}
We give a parameterized reduction from \textsc{$k$-Clique}: given a graph $G=(V,E)$ and an integer $k$,  determine whether there exists a clique of size $k$ in $G$.
The problem is  W[1]-complete parameterized by $k$ and admits no $f(k)n^{o(k)}$-time algorithm unless ETH fails~\cite{Downey1995,CHKX2004}.

Given an instance $(G=(V,E), k)$ of \textsc{$k$-Clique}, we construct \rerevised{a} seat graph $G'$ that consists of a clique $\{w_1,w_2, \ldots, w_k\}$ of size $k$ and $n-k$ isolated vertices $w_{k+1}, \ldots, w_n$.
Clearly, the size  $\gamma$ of a minimum vertex cover of $G'$ is $k-1$.
Let $\agents=V$.
Then we set the preferences of any pair of agents $u,v\in \agents$ by $f_u(v)=f_v(u)=1$ if $\{u,v\}\in E$, and otherwise $f_u(v)=f_v(u)=0$. Note that the preferences are binary.
 
Finally, we show that \textsc{$k$-Clique} is a yes-instance if and only if there exists an arrangement $\pi$ with social welfare $k(k-1)$ in $G'$.
Given an instance $(G,k)$ of  \textsc{$k$-Clique}, we give indices to each vertex $v_1,v_2, \ldots, v_n$ arbitrarily.  
Given a $k$-clique, we denote it by $C=\{v_1,v_2,\ldots,v_k\}$ without loss of generality.
Then we set $\pi(v_i)=w_i$ for any $i\in \{1,\ldots, n\}$.
Since $\{u,v\}\in E$ for any pair of $u,v\in C$, we have $U_{v_i}(\pi)=k-1$ for $i\in \{1,\ldots, k\}$.
For each $i\in \{k+1,\ldots, n\}$, $w_i$ is an isolated vertex, and hence $U_{v_i}(\pi)=0$.
Therefore, $\sw(\pi)=k(k-1)+0=k(k-1)$.

For the reverse direction, we are given an arrangement $\pi$ with social welfare $k(k-1)$.
Since $w_{k+1},\ldots, w_n$ are isolated vertices, $U_{\pi^{-1}(w_i)}(\pi)=0$ for $i\in \{k+1,\ldots, n\}$.
Moreover, because the preferences are binary, $U_{\pi^{-1}(w_i)}(\pi)\le k-1$ for $i\in \{1,\ldots, k\}$.
Thus, any agent $p=\pi^{-1}(w_i)$ for $i\in \{1,\ldots, k\}$ satisfies that $U_{\pi^{-1}(w_i)}(\pi)=k-1$ in order to achieve $\sw(\pi)=k(k-1)$.
This implies that $\{\pi^{-1}(w_i), \pi^{-1}(w_j)\}\in E$ for any pair of $w_i, w_j$ where $i, j\in \{1,\ldots, k\}$.
Therefore, $\{ \pi(w_1),\ldots,\pi(w_k) \}$ is a clique of size $k$.
\end{proof}
%
For {\sf EGA}, we show that it is weakly NP-hard even on a graph of vertex cover number 2, which again implies that it does not admit any parameterized algorithms by vertex cover number unless P$=$NP.
We give a reduction from \textsc{Partition}: given a finite set of integers $A=\{a_1,a_2,\ldots, a_{n}\}$ and $W=\sum_{i=1}^{n} a_i$, determine whether there is partition $(A_1, A_2)$ of $A$ where $A_1\cup A_2=A$ and  $\sum_{a\in A_1}a = \sum_{a\in A_2}a=W/2$.
The problem is weakly NP-complete \cite{GJ1979}.
\begin{theorem}\label{thm:EGA:vc2}
{\sf EGA} is weakly NP-hard even on a graph with $\vc=2$.
\end{theorem}
\begin{proof}
We are given a set of integers $A=\{a_1,a_2,\ldots, a_{n}\}$.
We define \revised{two sets} of agents  ${\bf A}=\{p_{a_1}, \ldots, p_{a_n}\}$ and ${\bf C}=\{c_1, c_2\}$.
Each agent in ${\bf A}$ corresponds to an element in $A$.
For $p_{a_i}\in {\bf A}$, we define $f_{p_{a_i}}(q)=W/2$ if $q\in {\bf C}$, and otherwise $f_{p_{a_i}}(q)=0$. \revised{Moreover, for $c\in {\bf C}$, we define $f_{c}(p_{a_i})=a_i$ if $p_{a_i}\in {\bf A}$, and otherwise $f_{c}(p_{a_i})=0$.}
Finally, we define the seat graph $G$ as a graph consisting of $S_1$ and $S_2$, where $S_i$ is a star of size $n/2+1$. Note that the vertex cover number of $G$ is 2.

In the following, we show that there is a partition $(A_1, A_2)$ where $\sum_{a\in A_1}a = \sum_{a\in A_2}a=W/2$ if and only if there is an arrangement such that the least utility is at least $W/2$ in $G$.
Given a partition $(A_1, A_2)$, let ${\bf A}_1$ and ${\bf A}_2$ be the corresponding sets of agents in ${\bf A}$. We assign agents in ${\bf A}_i$ to leaves of $S_i$ and $c_i$ to the center of $S_i$ for $i\in \{1,2\}$. 
In the arrangement, each utility is $W/2$.

Conversely, we are given an arrangement $\pi$ such that the least utility is at least $W/2$.
If $p\in {\bf A}$ is assigned to the center of a star, at least one agent in ${\bf A}$ is adjacent to only $p$. Then its utility is $0$.
Thus, $c_i\in {\bf C}$ must be assigned to the center of $S_i$.
By the definition of the preferences, the utilities of $c_1$ and $c_2$ are exactly $W/2$.
Thus, two sets of agents in the leaves of $S_1$ and $S_2$ correspond to $A_1$ and $A_2$.
\end{proof}

Similarly, we show that (symmetric) {\sf EFA} is weakly NP-hard even on a graph with vertex cover number 2.
The reduction is from \textsc{Partition} and the reduced graph is the same as the one in the proof of Theorem \ref{thm:EGA:vc2}.
The preferences are defined as $f_{p}(q)=f_{q}(p)=a_i$ if $p=p_{a_i}\in {\bf A}$ and $q\in {\bf C}$, and otherwise $f_p(q)=f_q(p)=0$.
\begin{theorem}\label{thm:EFA:vc2}
{\sf EFA} is weakly NP-hard even if the preferences are symmetric and the vertex cover number of the seat graph is 2.
\end{theorem}




\section{Parameterized Complexity of Local Search}\label{sec:localsearch}
As mentioned in Section \ref{subsec:motivation}, finding a stable solution under symmetric preferences by swapping two agents iteratively is PLS-complete.
In this section, we investigate the parameterized complexity of local search of \textsc{Stable Arrangement} by considering {\sf Local $k$-STA}, which determines whether a stable arrangement can be obtained from any given arrangement by $k$ swaps.

Given a set $\mathbf{P}$ of agents  with preferences $\preferencefamily_\agents$, a graph $G$ with $|V (G)| = |\mathbf{P}|$, an
arrangement $\pi : \mathbf{P} \rightarrow V (G)$, and an integer $k$, \textsc{\problemname}  asks whether  there is a stable arrangement $\pi'$
that can be obtained from $\pi$ \rerevised{within} $k$ swaps.
\begin{theorem}\label{thm:w1:hard}
\textsc{\problemname} is W[1]-hard parameterized by $k$ even if the preferences are symmetric.
\end{theorem}
\begin{proof}
	\def\agentsbuf{\agents}
	\renewcommand\agents{\mathbf{P}}
	We give a reduction from the \textsc{Independent Set} problem, which is known to be W[1]-complete \cite{Downey1995}. Let $(H, k)$ be an instance of \textsc{Independent Set} where $H$ is a graph \rerevised{with} $n$ vertices and \rerevised{an integer $k$ is} the parameter, and the question is whether $H$ has an independent set of size $k$. Throughout the following, for convenience, we will assume that $n > k + 2$.
	
	We will construct an instance $\cI = (\agents, \preferencefamily_\agents, G, \arrangement, k)$ of \textsc{\problemname} such that $\cI$ is a \yes{}-instance if and only if $H$ has an independent set of size $k$.
	
	We construct a set  $\agents$ of $n + 3k + 5$ agents, which is partitioned into the following subsets: 
	$\agents = \bfC_1 \cup \bfC_2 \cup \bfV \cup \{\bfx_1\} \cup \bfY \cup \{{\bfx}_2\}$.
	We have that $\card{\bfC_1} = k$, $\card{\bfC_2} = k + 2$, and $\card{\bfY} = k + 1$. Finally, $\bfV = V(H)$ and we may refer to elements of the set $V(H) = \bfV$ both as vertices of $H$ and of agents of $\agents$.
The definition of the preferences $\preferencefamily_\agents$ is given in Table~\ref{tab:preferences}. Note that they are symmetric.
\newcommand\edgecode{
	$\left\lbrace\begin{array}{ll} 
			-\largeval, &\mbox{ if } \agent\agentt\in E(H) \\
			0, &\mbox{ otherwise}
		\end{array}\right.$
}
%
%
\begin{table*}
	\centering
		\caption{The preferences $\preferencefamily_\agents$ given in the proof of Theorem~\ref{thm:w1:hard}. For $\agent, \agentt$ from the corresponding sets, the entry shows $\preference_\agent(\agentt)$.}
	\begin{tabular}{c|c|cccccc}
			& & & & {\large $q \in{}$} & & \\
			 \hline
			& & $\bfC_1$ & $\bfC_2$ & $\bfV$ & $\bfY$ & $\{\bfx_1\}$ & $\{\bfx_2\}$ \\
			\hline
			& $\bfC_1$ & $0$ & $-\largeval$ & $-1$ & $0$ & $1$ & $-1$ \\
			& $\bfC_2$ & $-\largeval$ & $0$ & $1$ & $1$ & $-\largeval$ & $-\largeval$ \\ 
			{\large $p \in{}$} & $\bfV$ & $-1$ & $1$ & \edgecode & $0$ & $-1$ & $0$ \\
			& $\bfY$ & $0$ & $1$ & $0$ & $0$ & $-\largeval$ & $1$ \\
			& $\{\bfx_1\}$ & $1$ & $-\largeval$ & $-1$ & $-\largeval$ & --- & $-\largeval$ \\
			& $\{\bfx_2\}$ & $-1$ & $-\largeval$ & $0$ & $1$ & $-\largeval$ & ---
	\end{tabular}
	\label{tab:preferences}
\end{table*}
	
	We construct a graph $G$ as follows. \rerevised{The graph} $G$ consists of one clique on $2k + 2$ vertices whose vertices are $C_1 \cup C_2$ with $\card{C_1} = k$ and $\card{C_2} = k+2$, one star on $n + 1$ vertices whose center is $x_1$ and whose leaves are called $V_H$, one star on $k + 2$ vertices whose center is $x_2$ and whose leaves are called $Y$.
	
	We now let $\arrangement$ be a map that maps bijectively $\bfC_1$ to $C_1$, $\bfC_2$ to $C_2$, $\bfV$ to $V_H$, $\bfY$ to $Y$, \revised{$\{\bfx_1\}$ to $x_1$ and $\{\bfx_2\}$ to $x_2$}. This finishes the construction of our instance $\cI = (\agents, \preferencefamily_\agents, G, \arrangement, k)$. See the instance in Figure~\ref{fig:swap-hardness}.


Before moving to the detailed proof, we give a high-level description of the reduction. 
In the reduction, the main part of the seat graph is the clique associated with $\bfC_1\cup \bfC_2$ and the star associated with $\{\bfx_1\}\cup \bfV$. Since agents in $\bfC_1$ and agents in $\bfC_2$ dislike each other, in order to make $\pi$ stable, we have to move them so that they are not adjacent. 
Since $|\bfC_2|=k+2$, 
the only way of possibly obtaining a stable arrangement within $k$ swaps is by moving agents in $\bfC_1$.
Suppose that there exists an independent set $S$ of size $k$ in $G$. Let $\bfS \subseteq \bfV$ be the set of agents corresponding to $S$. Then the arrangement obtained from $\pi$ after $k$ swaps of  $\bfC_1$ and $\bfS$ is stable because agents in $\bfx_1$ and $\bfC_1$ like each other, and agents in $\bfS$ and agents in $\bfC_2$ also like each other. Moreover, since $S$ is an independent set, agents in $\bfS$ are indifferent. 
Conversely, if there is no independent set of size $k$, we cannot swap agents in $\bfC_1$ and agents in $\bfV$ because agents in $\bfV$ dislike each other. Eventually, we can conclude that a stable arrangement cannot be obtained after $k$ swaps. Here, agents in $\{\bfx_2\}\cup \bfY$ play essential roles by joining a blocking pair so that any arrangement after $k$ swaps is not stable. In what follows, we give the detailed proof.


%
\newcommand\instance{(\agents, \preferencefamily_\agents, G, \arrangement, k)}
\begin{figure}
	\centering
	\includegraphics[height=.30\textheight]{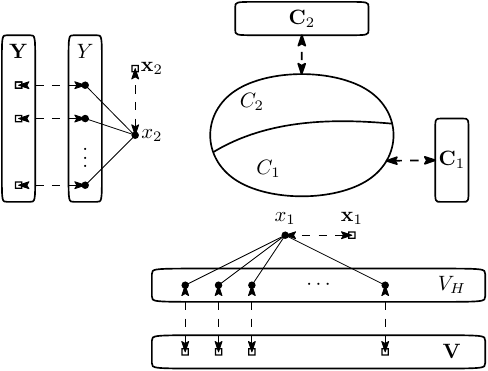}
	\caption{Illustration of the instance of \textsc{\problemname} constructed in the reduction of Theorem~\ref{thm:w1:hard}.}
	\label{fig:swap-hardness}
\end{figure}	
	
\begin{nestedclaim}\label{claim:cor:forward}
	If $H$ has an independent set $S$ of size $k$, then $\instance$ is a \yes{}-instance.
\end{nestedclaim}
\begin{claimproof}
Let $\arrangement'$ be the arrangement obtained from $\arrangement$ by swapping the assignments of the agents in $\bfS$ with the assignments of the agents in $\bfC_1$.
	Since $\card{S} = k$, it is clear that $\arrangement'$ can be reached from $\arrangement$ by $k$ swaps. We show that $\arrangement'$ is a stable arrangement. Suppose not, then there is a blocking pair $(\agent, \agentt)$ for $\arrangement'$, i.e.,\ if $\arrangement''$ is the $(\agent, \agentt)$-swap arrangement of $\arrangement'$, then $\utility_\agent(\arrangement'') > \utility_\agentt(\arrangement')$ and $\utility_\agentt(\arrangement'') > \utility_\agentt(\arrangement')$.
	
	We conduct a case analysis on which sets of the above described partition of $\agents$ contain $\agent$ and $\agentt$. 
	Throughout the following, we denote by $\arrangement''$ a $(\agent, \agentt)$-swap arrangement of $\arrangement'$ where $\agent$ and $\agentt$ are defined depending on the below cases.
\begin{description}
	\item[Case 1 ($\agent \in \bfC_1$).] Note that $\arrangement'(\agent) \in V_H$: in $\arrangement$, the agents in $\bfC_1$ were mapped to $C_1$ and $\arrangement'$ is obtained from $\arrangement$ by swapping the assignment of the agents in $\bfC_1$ with the agents in $\bfS \subseteq \bfV$ which are mapped to $V_H$. This implies that $\utility_\agent(\arrangement') = 1$. 
	As $\bfx_1$ is the only agent to which $\agent$ has a positive preference, this is the maximum utility of $\agent$ among all arrangements.
	%
	\item[Case 2 ($\agent \in \bfC_2$).] Note that $\arrangement'(\agent) \in C_2$. Furthermore, $\arrangement'(\agent)$ has $k$ neighbors to which an element of $\bfV$ is mapped, and $k+1$ neighbors to which an element of $\bfC_2$ is mapped. Hence, $\utility_\agent(\arrangement') = k$. We have that for any $\agentt \in \agents \setminus (\bfV \cup \bfY \cup \{\agent \})$, $\preference_\agent(\agentt) \le 0$, 
	and for $\agentt \in \bfV \cup \bfY$, $\preference_\agent(\agentt) = 1$,	
	so in order to increase the utility of $\agent$ in any swap arrangement of $\arrangement'$, we would have to map $\agent$ to a vertex that has more than $k + 1$ neighbors to which an element of $\bfV \cup \bfY$ is mapped.
	
	There are two such possibilities, namely either $\agentt = \bfx_1$ or $\agentt = \bfx_2$. Suppose $\agentt = \bfx_1$. Note that $\arrangement'(\bfx_1) = x_1$, and $x_1$ has $k$ neighbors to which an agent from $\bfC_1$ is mapped, and $n - k$ neighbors to which an agent from $\bfV$ is mapped. Hence, $\utility_{\agentt}(\arrangement') = k - (n - k) = 2k - n$. As $\arrangement''(q) \in C_2$, $\utility_q(\arrangement'') = -(k+1)\largeval-k= -(k+2)\largeval$. 
	As $\utility_\agentt(\arrangement') = 2k - n > -(k+2)\largeval = \utility_\agentt(\arrangement'')$,
	$(\agent, \agentt)$ is not a blocking pair for $\arrangement'$.
	If $\agentt = \bfx_2$, then we observe that $\utility_\agentt(\arrangement') = k + 1 > - (k+1)\largeval = \utility_\agentt(\arrangement'')$,
	so $(\agent, \agentt)$ is not a blocking pair for $\arrangement'$ either.
	\item[Case 3.1 ($\agent \in \bfV$, \rerevised{$\agent \in \bfS$}).] In this case, $\arrangement'(\agent) \in C_1$, and $\utility_\agent(\arrangement') = k + 2$: $\arrangement'(\agent)$ has $k + 2$ neighbors to which agents in $\bfC_2$ are mapped, and to the remaining neighbors of $\arrangement'(\agent)$, \rerevised{agents in $\bfS$} are mapped. Since $S$ is an independent set, the latter contributes with a total value of $0$ to $\utility_\agent(\arrangement')$. We have that $\arrangement'$ in fact achieves the maximum utility for $\agent$, among all arrangements: The only agents to which $\agent$ has a positive preference are the ones in $\bfC_2$, and in $\arrangement'$, all agents in $\bfC_2$ are mapped to neighbors of $\arrangement'(\agent)$.
	\item[Case 3.2 ($\agent \in \bfV$, $\agent \notin \rerevised{\bfS}$).] 
	We observe that if $\agentt \in \bfC_1 \cup (\bfV \setminus \rerevised{\bfS})$, then the utility of $\agent$ compared to the resulting swap arrangement does not change, since in this case, $\arrangement'(\agent) \in V_H \ni \arrangement'(\agentt)$. The remaining ones are as follows:
	\begin{itemize}
		\item If $\agentt \in \rerevised{\bfS}$, then $\utility_\agentt(\arrangement') = k + 2 > - 1 = \utility_\agentt(\arrangement'')$.
		\item If $\agentt \in {\bfC}_2$, then $\utility_\agentt(\arrangement') = k > -\largeval = \utility_\agentt(\arrangement'')$.
		\item If $\agentt \in \bfY$, then $\utility_\agentt(\arrangement') = 1 > -\largeval = \utility_\agentt(\arrangement'')$.
		\item If $\agentt = \bfx_1$, then $\utility_\agent(\arrangement') = - 1 > -k\ge \utility_\agent(\arrangement'')$.
		\item If $\agentt = {\bfx}_2$, then $\utility_\agentt(\arrangement') = k + 1 > - \largeval = \utility_\agentt(\arrangement'')$.
	\end{itemize}
	\item[Case 4 ($\agent \in \bfY$).] We have that $\utility_\agent(\arrangement') = 1$, as $x_2 = \arrangement'^{-1}(\bfx_2)$ is the only neighbor of $\arrangement'(\agent)$. 
	There is only one way to increase the utility of $\agent$ in one swap. It is to map $\agent$ to a vertex to whose neighbors agents of $\bfC_2$ are mapped. This means that we would have to swap $\agent$ with some $\agentt \in \bfC_2 \cup \rerevised{\bfS}$. 
	However, if $q\in \bfC_2$, we would have that $\utility_\agentt(\arrangement')=k>-n=\utility_\agentt(\arrangement'')$ in the $(p,q)$-swap arrangement $\arrangement''$ for $\arrangement'$. Moreover, if $q\in \rerevised{\bfS}$, then $\arrangement'(q)\in C_1$. Then we would have that $\utility_\agentt(\arrangement')=k+2>0=\utility_\agentt(\arrangement'')$ in the $(p,q)$-swap arrangement $\arrangement''$.
	%
	
	\item[Case 5 ($\agent = \bfx_1$).] 
	We observe that $\utility_{\bfx_1}(\arrangement') = 2k-n$. In the $(p,q)$-swap arrangement $\arrangement''$, we have: 
	\begin{itemize}
		\item If $\agentt\in \rerevised{\bfS}$, $\utility_\agent(\arrangement') = 2k-n >-n(k+2)-(k-1) =\utility_\agent(\arrangement'')$.
		\item If $\agentt\in {{\bf V}\setminus \rerevised{\bfS}}$, $\utility_\agentt(\arrangement') = -1>-k-1 \ge \utility_\agentt(\arrangement'')$.
		\item If $\agentt\in {\bfC_1}$, $\utility_\agentt(\arrangement') = 1> 1 - (n-k) = \utility_\agentt(\arrangement'')$.
		\item If $\agentt\in {\bfC_2}$, $\utility_\agent(\arrangement') = 2k-n >-n(k+1)-k=\utility_\agent(\arrangement'')$.
				\item If $\agentt\in {\bf Y}$, $\utility_\agent(\arrangement')=2k-n  > -n=\utility_\agent(\arrangement'')$.
		\item If $\agentt = {\bfx_2}$, $\utility_\agent(\arrangement') =2k-n >-n(k+1)=\utility_\agent(\arrangement'')$.
	\end{itemize}

\item[Case 6 ($\agent =  \bfx_2$).] We observe that  $\utility_{\bfx_2}(\arrangement') = k+1$, and this is the maximum utility that $\bfx_2$ can achieve.
\end{description}
We have shown that there is no blocking pair in $\arrangement'$, so it is indeed a stable arrangement.
\end{claimproof}

We now prove the converse direction of the correctness of this reduction. Suppose $\instance$ is a \yes{}-instance of \textsc{\problemname}, and let $\arrangement'$ be an arrangement obtained from $\arrangement$ in $k$ swaps such that $\arrangement'$ is stable. 
We will show that for $\arrangement'$ to be stable, these $k$ swaps need to swap the agents of $\bfC_1$ with $k$ agents \rerevised{in $\bfS\subseteq \bfV$}, and the agents in $\rerevised{\bfS}$ need to correspond to an independent set in $H$.

First, we immediately observe that since $\card{\bfC_2} = k + 2$, and for each $\agent \in \bfC_2$, $\arrangement(\agent) \in C_2$, at least two agents of $\bfC_2$ remain assigned to a vertex in $C_2$ by $\arrangement'$. For a similar reason, namely $\card{\bfY} = k + 1$, at least one agent in $\bfY$ is still mapped to a vertex in $Y$.
\begin{nestedobservation}\label{obs:at:least:2:C2}\label{obs:at:least:one:Y}
	There are two distinct $\agent_1, \agent_2 \in \bfC_2$ such that $\arrangement'(\agent_1) \in C_2$ and $\arrangement'(\agent_2) \in C_2$. Furthermore, there is at least one agent $\agent \in \bfY$ such that $\arrangement'(\agent) \in Y$.
\end{nestedobservation}

Next, we show that in $\arrangement'$, neither $\bfx_1$ nor $\bfx_2$ can be mapped to a vertex in the clique $C_1 \cup C_2$.

\begin{nestedclaim}\label{claim:position:x1}
	$\arrangement'(\bfx_1) \notin C_1 \cup C_2$ and $\arrangement'(\bfx_2) \notin C_1 \cup C_2$.
\end{nestedclaim}
\begin{claimproof}
	Suppose $\arrangement'(\bfx_1) \in C_1 \cup C_2$. First, we have that $\utility_{\bfx_1}(\arrangement') \le k - 2\largeval$, since $k$ is the maximum utility that $\bfx_1$ can have in any arrangement, and by Observation~\ref{obs:at:least:2:C2}, there are at least two agents from $\bfC_2$ that are mapped to neighbors of $\arrangement'(\bfx_1)$. Again by Observation~\ref{obs:at:least:one:Y}, there is one agent $y \in \bfY$ that is still mapped to a vertex in $Y$. That vertex has only one neighbor (the center of the star to which the vertices in $Y$ are leaves), so $\utility_y(\arrangement') \le 1$.
	Now, if $\arrangement''$ is the $(\bfx_1, y)$-swap arrangement of $\arrangement'$, then $\utility_{\bfx_1}(\arrangement'') \ge - \largeval$ ($\arrangement''(\bfx_1)$ having only one neighbor) and $\utility_y(\arrangement'') \ge 2$: $\bfx_1$ is not mapped to a neighbor of $y$ in $\arrangement''$, and $y$ has non-negative preferences to the remaining agents; by Observation~\ref{obs:at:least:2:C2}, $\arrangement''(y)$ has at least two neighbors to which an agent of $\bfC_2$ is mapped. So, $\utility_{\bfx_1}(\arrangement'') \ge -\largeval > k - 2\largeval \ge \utility_{\bfx_1}(\arrangement')$, as $n > k$, and $\utility_y(\arrangement'') \ge 2 > 1 \ge \utility_y(\arrangement')$,
	hence $(\bfx_1, y)$ is a blocking pair in $\arrangement'$, a contradiction to $\arrangement'$ being stable.
	
	Similarly, if $\arrangement'(\bfx_2) \in C_1 \cup C_2$ and $\arrangement''$ is the $(\bfx_2, y)$-swap arrangement of $\arrangement'$, then we have that $\utility_{\bfx_2}(\arrangement'') \ge -\largeval > k - 2\largeval \ge \utility_{\bfx_2}(\arrangement')$ and $\utility_y(\arrangement'') \ge 2 > 1 \ge \utility_y(\arrangement')$, so $(\bfx_2, y)$ is a blocking pair.
\end{claimproof}

\begin{nestedclaim}\label{claim:position:C1}
	There is no agent $\agent \in \bfC_1$ such that $\arrangement'(\agent) \in C_1 \cup C_2$.
\end{nestedclaim}
\begin{claimproof}
	Suppose there is. Then we have that $\utility_{\agent}(\arrangement') \le - 2\cdot \largeval$, since by Observation~\ref{obs:at:least:2:C2}, $\agent$ has at least two neighbors to which an agent of $\bfC_2$ is mapped, 
	and $\bfx_1$ is the only agent that $\agent$ has a positive preference to, and $\bfx_1$ is not assigned to a neighbor of $\arrangement'(\agent)$ by Claim~\ref{claim:position:x1}.
	
	By Observation~\ref{obs:at:least:one:Y}, there is an agent $y \in \bfY$ with $\arrangement'(y) \in Y$.
	We have that $\utility_y(\arrangement') \le 1$.
	Now, $(\agent, y)$ is a blocking pair of $\arrangement'$: let $\arrangement''$ be the $(\agent, y)$-swap arrangement of $\arrangement'$. Then, $\utility_\agent(\arrangement'') \ge -n$, since $\arrangement''(\agent) \in Y$ which only has one neighbor $x_2$. Furthermore we can observe that $\utility_y(\arrangement'') \ge 2$: $\arrangement''(y)$ has at least two neighbors to which agents of $\bfC_2$ are mapped by Observation~\ref{obs:at:least:2:C2}, $\bfx_1$ is the only agent to which $y$ has a negative preference, and Claim~\ref{claim:position:x1} ensures that $\arrangement''(\bfx_1) = \arrangement'(\bfx_1) \notin C_1 \cup C_2$. 
	Hence,
	$\utility_\agent(\arrangement'') \ge - n> - 2\cdot \largeval \ge \utility_\agent(\arrangement')$, and
	$\utility_y(\arrangement'') \ge 2 > 1 \ge \utility_y(\arrangement')$, \rerevised{a contradiction with $\arrangement'$ being stable.}
\end{claimproof}

Claim~\ref{claim:position:C1} yields the following information about the swaps that were executed to obtain $\arrangement'$ from $\arrangement$.
\begin{nestedclaim}\label{claim:swaps}
	The $k$ swaps executed to obtain $\arrangement'$ from $\arrangement$ are $(\agent_1, \bfc_1)$, $(\agent_2, \bfc_2)$, $\ldots$, $(\agent_k, \bfc_k)$, where
	\begin{enumerate}
		\item $\bfC_1 = \{\bfc_1, \ldots, \bfc_k\}$, and for all $i \in [k]$, $\agent_i \notin \bfC_1 \cup \bfC_2$.\label{claim:swaps:C1}
		\item for all $i, j \in [k]$ with $i \neq j$, $\agent_i \neq \agent_j$.\label{claim:swaps:other}
	\end{enumerate}
\end{nestedclaim}
\begin{claimproof}
	Part~\ref{claim:swaps:C1} is immediate from Claim~\ref{claim:position:C1} and the fact that $\card{\bfC_1} = k$: if there was one swap that did not remove an agent in $\bfC_1$ from $C_1 \cup C_2$, then for at least one agent $\bfc \in \bfC_1$, $\arrangement'(\bfc) \in C_1 \cup C_2$. 
	For Part~\ref{claim:swaps:other}, suppose there are $i, j \in [k]$ with $i \neq j$ such that $\agent_i = \agent_j \eqdef \agent$, and suppose wlog.\ that $i < j$. Then, after swapping \rerevised{$\bfc_i$} and $\agent$, we have that $\agent$ is mapped to a vertex in $C_1$. Then, when swapping $\agent$ and \rerevised{$\bfc_j$}, \rerevised{$\bfc_j$} is still mapped to a vertex in $C_1$. Since by Part~\ref{claim:swaps:C1}, each agent is affected by at most one swap, we have that $\arrangement'(\rerevised{\bfc_j}) \in C_1 \cup C_2$, a contradiction with Claim~\ref{claim:position:C1}.
\end{claimproof}

Now, combining Claim~\ref{claim:swaps} with \ref{claim:position:x1} tells us that $\bfx_1$ and $\bfx_2$ remain unaffected by the $k$ swaps that yielded $\arrangement'$.
\begin{nestedobservation}\label{obs:fix:x1:x2}
	$\arrangement'(\bfx_1) = x_1$ and $\arrangement'(\bfx_2) = x_2$.
\end{nestedobservation}

Claim~\ref{claim:swaps} ensures that each agent is affected by at most one swap and that each swap affects one unique agent in $\bfC_1$. Furthermore, it rules out that the agents of $\bfC_1$ are swapped with agents from $\bfC_1 \cup \bfC_2$, and Observation~\ref{obs:fix:x1:x2} rules out that they are swapped with $\bfx_1$ or $\bfx_2$. As our goal is to show that they are swapped with agents from $\bfV$, the only case that remains to be ruled out is when they are swapped with an agent from $\bfY$.
\begin{nestedclaim}\label{claim:pi:Y}
	There is no $i \in [k]$ such that $\agent_i \in \bfY$.
\end{nestedclaim}
\begin{claimproof}
Suppose there is, and let $(\agent_i, \bfc_i)$ be the corresponding swap. Let $\arrangement^*$ be the $(\agent_i, \bfc_i)$-swap arrangement of $\arrangement$ and note that by Claim~\ref{claim:swaps} and Observation~\ref{obs:fix:x1:x2}, $\utility_{\bfc_i}(\arrangement^*) = \utility_{\bfc_i}(\arrangement') = -1$. As $\card{\bfV} > k$, there is at least one agent {$\agentt \in {\bf V}$} that $\arrangement'$ assigns to $V_H$. Again by Observation~\ref{obs:fix:x1:x2}, we have that $\utility_\agentt(\arrangement') = -1$. Let $\arrangement''$ be the $(\bfc_i, \agentt)$-swap arrangement of $\arrangement'$. Then, $\utility_{\bfc_i}(\arrangement'') = 1 > - 1 = \utility_{\bfc_i}(\arrangement')$ and $\utility_\agentt(\arrangement'') = 0 > - 1 = \utility_\agentt(\arrangement')$, so $(\bfc_i, \agentt)$ is a blocking pair for $\arrangement'$, a contradiction with $\arrangement'$ being stable.
\end{claimproof}

%
\begin{nestedclaim}
	If $\arrangement'$ is stable, then $H$ contains an independent set of size $k$.
\end{nestedclaim}
\begin{claimproof}
	The above claims and observations lead us to the conclusion that the agents $\agent_1, \ldots, \agent_k$ (in the notation of Claim~\ref{claim:swaps}) form a size-$k$ subset, say $S$, of $\bfV = V(H)$. We argue that if $\arrangement'$ is stable, then $S$ is indeed an independent set in $H$. Suppose for the sake of a contradiction that there is an edge between $\agent_i$ and $\agent_j$ in $H$ for some $i \neq j$. Then, $\utility_{\agent_i}(\arrangement') \le (k + 2) - \largeval$. Consider again the agent $y \in \bfY$ from Observation~\ref{obs:at:least:one:Y} which is such that $\arrangement'(y) \in Y$. We have that $\utility_y(\arrangement') = 1$ (together with Observation~\ref{obs:fix:x1:x2}). Now let $\arrangement''$ be the $(\agent_i, y)$-swap arrangement of $\arrangement'$. Then, $\utility_{\agent_i}(\arrangement'') = 0 > (k+2) - \largeval = \utility_{\agent_i}(\arrangement')$ (as by our initial assumption, $n > k + 2$) and $\utility_y(\arrangement'') = k + 2 > 1 = \utility_y(\arrangement')$, so $(\agent_i, y)$ is a blocking pair for $\arrangement'$.
\end{claimproof}

This concludes the correctness proof of the reduction. We observe that $\card{\agents} = n + 3k + 5 = \cO(n)$, $\card{\preferencefamily_\agents} = \cO(n^2)$, $\card{V(G)} = \card{\agents} = \cO(n)$, and $\card{E(G)} = \cO(n + k^2)$ ($G$ contains two stars with $n$ and $k+1$ leaves, respectively and a clique on $\cO(k)$ vertices). So, the size of the instance of \textsc{\problemname} is $\cO(n^2)$ and the parameter $k$ remained unchanged, which completes the proof.
\renewcommand\agents{\agentsbuf}
\end{proof}

We observe that the structure of the instance of \textsc{\problemname} is in fact quite restricted and yields the following stronger form of Theorem~\ref{thm:w1:hard}. First, in the preferences of the \textsc{\problemname} instance, we only have $4$ different values of preferences. Second, the seat graph $G$ obtained in the reduction above has a vertex cover of size $2k + 3$: take $2k+1$ vertices from $C_1 \cup C_2$ and the vertices $x_1$ and $x_2$. This means that even including the vertex cover number $\vc$ of the seat graph in the parameter, the problem remains W[1]-hard.
\begin{corollary}
	\textsc{\problemname} remains W[1]-hard parameterized by $k + \vc$ where $\vc$ denotes the vertex cover number of the seat graph, even when the number of preference values is $4$.
\end{corollary}

Moreover, it is not difficult to see that \textsc{\problemname} can be solved in time $n^{\cO(k)}$ by brute force. We simply guess all sets of $k$ pairs of agents, swap their assignments, and then verify whether or not the resulting assignment has a blocking pair. On the other hand, the value of the parameter $k + \vc$ in the above reduction is \emph{linear} in the value of the parameter of the \textsc{Independent Set} instance.
Since \textsc{Independent Set} does not have an $n^{o(k)}$ time algorithm unless \ETH{} fails \cite{CHKX2004}, this implies that the runtime of this naive brute force algorithm is in some sense tight under  \ETH{} --- even when the vertex cover number of the seat graph can be considered another component of the parameter, and even when we only have $4$ different choices for values of preferences.
\begin{corollary}
	\textsc{\problemname} can be solved in time $n^{\cO(k)}$, and there is no $n^{o(k + \vc)}$ time algorithm, where $\vc$ denotes the vertex cover number of the seat graph, even when the number of preference values is $4$, unless \ETH{} is false.
\end{corollary}
\section{Conclusion}
In this paper, we embark on a new model of hedonic games, called \textsc{Seat Arrangement}. The proposed model is powerful enough to treat real-world topological preferences. The results of the paper are summarized as follows: (1) We obtained basic results for stability and fairness. In particular, we proved that the $\pof$ is unbounded for the nonnegative case and we gave an upper bound  $\tilde{d}(G)$  and an almost tight lower bound $\tilde{d}(G)-1/4$ for the binary case. 
(2) We presented the dichotomies of the computational complexity of four \textsc{Seat Arrangement} problems in terms of the order of components. (3) We proved that {\sf UTA} can be solved in time $n^{O(\vc)}$ where $\vc$ is the vertex cover number whereas it is W[1]-hard for $\vc$ and cannot be solved in time $n^{o(n)}$ and $f(\vc)n^{o(\vc)}$, respectively, under ETH. Furthermore, {\sf EGA} and symmetric {\sf EFA} are weakly NP-hard even on graphs with $\vc=2$.
(4) We proved that {\sf Local $k$-STA} is W[1]-hard when parameterized by $k+\vc$ and cannot be solved in time $n^{o(k+\vc)}$ under $\ETH$, whereas it can be solved in time $n^{O(k)}$.

Finally, we give further directions for this research. First, our complexity analysis of \textsc{Seat Arrangement} is fundamental. Since \textsc{Seat Arrangement} has various input parameters and desirable solution concepts, it would be worth investigating the more detailed analysis of the problem.
After the extended abstract of this paper was published, several papers indeed studied the \textsc{Seat Arrangement} model. For example, Ceylan, Chen, and Roy conducted a more detailed study of the parameterized complexity of \textsc{Seat Arrangement} by considering the parameters: the number of non-isolated vertices and the maximum degree of the seat graph~\cite{SA:Ceylan0R23}. Berriaud, Constantinescu, and  Wattenhofe discussed the complexity of \textsc{Seat Arrangement} on very restricted graph classes such as cycles and paths when the number of classes of agents is bounded or the preferences are restricted~\cite{SA:BerriaudCW23}. 
In line with these studies, it is intriguing to consider extending the tractable or intractable cases for \textsc{Seat Arrangement}. Also, designing approximation algorithms is another approach to addressing the problem.

Another direction for \textsc{Seat Arrangement} is considering other solution concepts. In the context of stable matching or hedonic games, a lot of solution concepts are proposed~\cite{Aziz2013}. Importing these solution concepts to \textsc{Seat Arrangement} would accelerate further research.
In this line, Wilczynski recently considered \textsc{Seat Arrangement} under \emph{ordinal} preferences and investigated the stability and popularity~\cite{SA:Wilczynski23}. 

Lastly, several variants of \textsc{Seat Arrangement} can be considered by changing the objective function.
When assessing the average happiness of an agent based on their neighbors, we can incorporate the notion of \emph{fractionality} from fractional hedonic games~\cite{Bilo2018,AzizBBHOP19,HanakaIO25} into \textsc{Seat Arrangement}. This would involve defining the fractional utility divided by the degree of the vertex to which the agent is assigned. As mentioned in the introduction, the fractional utility is also used for  \textsc{Schelling Games}~\cite{SG:AgarwalEGISV21,SG:ChauhanLM18,SG:BiloBLM22,KreiselBFN24}.
Furthermore, we can take into account the influence of neighbors' neighbors in  \textsc{Seat Arrangement}. Such a consideration can be addressed by using the notion of social distance games~\cite{Branzei2011,Okubo2019,BullingerS24}.  

In this way, there are various directions for \textsc{Seat Arrangement}, and a diverse range of research can be expected in future work.


\bibliographystyle{plainurl}
\bibliography{ref}

\begin{thebibliography}{10}

\bibitem{SG:AgarwalEGISV21}
Aishwarya Agarwal, Edith Elkind, Jiarui Gan, Ayumi Igarashi, Warut Suksompong,
  and Alexandros~A. Voudouris.
\newblock Schelling games on graphs.
\newblock {\em Artif. Intell.}, 301:103576, 2021.
\newblock \href {https://doi.org/10.1016/j.artint.2021.103576}
  {\path{doi:10.1016/j.artint.2021.103576}}.

\bibitem{Alcalde1994}
Jos{\'e} Alcalde.
\newblock Exchange-proofness or divorce-proofness? {S}tability in one-sided
  matching markets.
\newblock {\em Economic Design}, 1(1):275--287, 1994.
\newblock \href {https://doi.org/10.1007/BF02716626}
  {\path{doi:10.1007/BF02716626}}.

\bibitem{AzizBBHOP19}
Haris Aziz, Florian Brandl, Felix Brandt, Paul Harrenstein, Martin Olsen, and
  Dominik Peters.
\newblock Fractional hedonic games.
\newblock {\em {ACM} Trans. Economics and Comput.}, 7(2):6:1--6:29, 2019.
\newblock \href {https://doi.org/10.1145/3327970} {\path{doi:10.1145/3327970}}.

\bibitem{Aziz2013}
Haris Aziz, Felix Brandt, and Hans~Georg Seedig.
\newblock Computing desirable partitions in additively separable hedonic games.
\newblock {\em Artif. Intell.}, 195:316--334, 2013.
\newblock \href {https://doi.org/10.1016/J.ARTINT.2012.09.006}
  {\path{doi:10.1016/J.ARTINT.2012.09.006}}.

\bibitem{Aziz2017}
Haris Aziz and Adrian Goldwaser.
\newblock Coalitional exchange stable matchings in marriage and roommate
  markets.
\newblock In Kate Larson, Michael Winikoff, Sanmay Das, and Edmund~H. Durfee,
  editors, {\em Proceedings of the 16th Conference on Autonomous Agents and
  MultiAgent Systems, {AAMAS} 2017, S{\~{a}}o Paulo, Brazil, May 8-12, 2017},
  pages 1475--1477. {ACM}, 2017.
\newblock URL: \url{http://dl.acm.org/citation.cfm?id=3091334}.

\bibitem{SA:BerriaudCW23}
Damien Berriaud, Andrei Constantinescu, and Roger Wattenhofer.
\newblock Stable dinner party seating arrangements.
\newblock In Jugal Garg, Max Klimm, and Yuqing Kong, editors, {\em Proceedings
  of the 19th International Conference on Web and Internet Economics, {WINE}
  2023, Shanghai, China, December 4-8, 2023, Proceedings}, volume 14413 of {\em
  Lecture Notes in Computer Science}, pages 3--20. Springer, 2023.
\newblock \href {https://doi.org/10.1007/978-3-031-48974-7_1}
  {\path{doi:10.1007/978-3-031-48974-7_1}}.

\bibitem{Bertsimas2011}
Dimitris Bertsimas, Vivek~F. Farias, and Nikolaos Trichakis.
\newblock The price of fairness.
\newblock {\em Oper. Res.}, 59(1):17--31, 2011.
\newblock \href {https://doi.org/10.1287/OPRE.1100.0865}
  {\path{doi:10.1287/OPRE.1100.0865}}.

\bibitem{SG:BiloBLM22}
Davide Bil{\`{o}}, Vittorio Bil{\`{o}}, Pascal Lenzner, and Louise Molitor.
\newblock Topological influence and locality in swap schelling games.
\newblock {\em Auton. Agents Multi Agent Syst.}, 36(2):47, 2022.
\newblock \href {https://doi.org/10.1007/s10458-022-09573-7}
  {\path{doi:10.1007/s10458-022-09573-7}}.

\bibitem{Bilo2018}
Vittorio Bil{\`{o}}, Angelo Fanelli, Michele Flammini, Gianpiero Monaco, and
  Luca Moscardelli.
\newblock Nash stable outcomes in fractional hedonic games: Existence,
  efficiency and computation.
\newblock {\em J. Artif. Intell. Res.}, 62:315--371, 2018.
\newblock \href {https://doi.org/10.1613/JAIR.1.11211}
  {\path{doi:10.1613/JAIR.1.11211}}.

\bibitem{Bod2019}
Hans~L. Bodlaender, Tesshu Hanaka, Yasuaki Kobayashi, Yusuke Kobayashi, Yoshio
  Okamoto, Yota Otachi, and Tom~C. van~der Zanden.
\newblock Subgraph isomorphism on graph classes that exclude a substructure.
\newblock {\em Algorithmica}, 82(12):3566--3587, 2020.
\newblock \href {https://doi.org/10.1007/S00453-020-00737-Z}
  {\path{doi:10.1007/S00453-020-00737-Z}}.

\bibitem{Bogomolnaia2002}
Anna Bogomolnaia and Matthew~O. Jackson.
\newblock The stability of hedonic coalition structures.
\newblock {\em Games Econ. Behav.}, 38(2):201--230, 2002.
\newblock \href {https://doi.org/10.1006/GAME.2001.0877}
  {\path{doi:10.1006/GAME.2001.0877}}.

\bibitem{Handbook2016}
Felix Brandt, Vincent Conitzer, Ulle Endriss, J{\'{e}}r{\^{o}}me Lang, and
  Ariel~D. Procaccia, editors.
\newblock {\em Handbook of Computational Social Choice}.
\newblock Cambridge University Press, 2016.
\newblock \href {https://doi.org/10.1017/CBO9781107446984}
  {\path{doi:10.1017/CBO9781107446984}}.

\bibitem{Branzei2009}
Simina Br{\^{a}}nzei and Kate Larson.
\newblock Coalitional affinity games and the stability gap.
\newblock In Craig Boutilier, editor, {\em {IJCAI} 2009, Proceedings of the
  21st International Joint Conference on Artificial Intelligence, Pasadena,
  California, USA, July 11-17, 2009}, pages 79--84, 2009.
\newblock URL: \url{http://ijcai.org/Proceedings/09/Papers/024.pdf}.

\bibitem{Branzei2011}
Simina Br{\^{a}}nzei and Kate Larson.
\newblock Social distance games.
\newblock In Toby Walsh, editor, {\em Proceedings of the 22nd International
  Joint Conference on Artificial Intelligence, {IJCAI} 2011, Barcelona,
  Catalonia, Spain, July 16-22, 2011}, pages 91--96. {IJCAI/AAAI}, 2011.
\newblock \href {https://doi.org/10.5591/978-1-57735-516-8/IJCAI11-027}
  {\path{doi:10.5591/978-1-57735-516-8/IJCAI11-027}}.

\bibitem{BullingerS24}
Martin Bullinger and Warut Suksompong.
\newblock Topological distance games.
\newblock {\em Theor. Comput. Sci.}, 981:114238, 2024.
\newblock \href {https://doi.org/10.1016/J.TCS.2023.114238}
  {\path{doi:10.1016/J.TCS.2023.114238}}.

\bibitem{Caragiannis2012}
Ioannis Caragiannis, Christos Kaklamanis, Panagiotis Kanellopoulos, and Maria
  Kyropoulou.
\newblock The efficiency of fair division.
\newblock {\em Theory Comput. Syst.}, 50(4):589--610, 2012.
\newblock \href {https://doi.org/10.1007/S00224-011-9359-Y}
  {\path{doi:10.1007/S00224-011-9359-Y}}.

\bibitem{Cechlarova2002}
Katar{\'{\i}}na Cechl{\'{a}}rov{\'{a}}.
\newblock On the complexity of exchange-stable roommates.
\newblock {\em Discret. Appl. Math.}, 116(3):279--287, 2002.
\newblock \href {https://doi.org/10.1016/S0166-218X(01)00230-X}
  {\path{doi:10.1016/S0166-218X(01)00230-X}}.

\bibitem{Cechlarova2005}
Katar{\'{\i}}na Cechl{\'{a}}rov{\'{a}} and David~F. Manlove.
\newblock The exchange-stable marriage problem.
\newblock {\em Discret. Appl. Math.}, 152(1-3):109--122, 2005.
\newblock \href {https://doi.org/10.1016/J.DAM.2005.06.003}
  {\path{doi:10.1016/J.DAM.2005.06.003}}.

\bibitem{SA:Ceylan0R23}
Esra Ceylan, Jiehua Chen, and Sanjukta Roy.
\newblock Optimal seat arrangement: What are the hard and easy cases?
\newblock In {\em Proceedings of the Thirty-Second International Joint
  Conference on Artificial Intelligence, {IJCAI} 2023, 19th-25th August 2023,
  Macao, SAR, China}, pages 2563--2571. ijcai.org, 2023.
\newblock \href {https://doi.org/10.24963/IJCAI.2023/285}
  {\path{doi:10.24963/IJCAI.2023/285}}.

\bibitem{SG:ChauhanLM18}
Ankit Chauhan, Pascal Lenzner, and Louise Molitor.
\newblock Schelling segregation with strategic agents.
\newblock In Xiaotie Deng, editor, {\em Proceedings of the 11th International
  Symposium on Algorithmic Game Theory, {SAGT} 2018, Beijing, China, September
  11-14, 2018}, volume 11059 of {\em Lecture Notes in Computer Science}, pages
  137--149. Springer, 2018.
\newblock \href {https://doi.org/10.1007/978-3-319-99660-8_13}
  {\path{doi:10.1007/978-3-319-99660-8_13}}.

\bibitem{CHKX2004}
Jianer Chen, Xiuzhen Huang, Iyad~A. Kanj, and Ge~Xia.
\newblock Linear {FPT} reductions and computational lower bounds.
\newblock In L{\'{a}}szl{\'{o}} Babai, editor, {\em Proceedings of the 36th
  Annual {ACM} Symposium on Theory of Computing, Chicago, IL, USA, June 13-16,
  2004}, pages 212--221. {ACM}, 2004.
\newblock \href {https://doi.org/10.1145/1007352.1007391}
  {\path{doi:10.1145/1007352.1007391}}.

\bibitem{Chen2010}
Jianer Chen, Iyad~A. Kanj, and Ge~Xia.
\newblock Improved upper bounds for vertex cover.
\newblock {\em Theor. Comput. Sci.}, 411(40-42):3736--3756, 2010.
\newblock \href {https://doi.org/10.1016/J.TCS.2010.06.026}
  {\path{doi:10.1016/J.TCS.2010.06.026}}.

\bibitem{SA:ChenCJS21}
Jiehua Chen, Adrian Chmurovic, Fabian Jogl, and Manuel Sorge.
\newblock On (coalitional) exchange-stable matching.
\newblock In Ioannis Caragiannis and Kristoffer~Arnsfelt Hansen, editors, {\em
  Proceedings of the 14th International Symposium on Algorithmic Game Theory,
  {SAGT} 2021, Aarhus, Denmark, September 21-24, 2021}, volume 12885 of {\em
  Lecture Notes in Computer Science}, pages 205--220. Springer, 2021.
\newblock \href {https://doi.org/10.1007/978-3-030-85947-3_14}
  {\path{doi:10.1007/978-3-030-85947-3_14}}.

\bibitem{Cygan2017}
Marek Cygan, Fedor~V. Fomin, Alexander Golovnev, Alexander~S. Kulikov, Ivan
  Mihajlin, Jakub Pachocki, and Arkadiusz Socala.
\newblock Tight lower bounds on graph embedding problems.
\newblock {\em J. {ACM}}, 64(3):18:1--18:22, 2017.
\newblock \href {https://doi.org/10.1145/3051094} {\path{doi:10.1145/3051094}}.

\bibitem{Cygan2015}
Marek Cygan, Fedor~V. Fomin, Lukasz Kowalik, Daniel Lokshtanov, D{\'{a}}niel
  Marx, Marcin Pilipczuk, Michal Pilipczuk, and Saket Saurabh.
\newblock {\em Parameterized Algorithms}.
\newblock Springer, 2015.
\newblock \href {https://doi.org/10.1007/978-3-319-21275-3}
  {\path{doi:10.1007/978-3-319-21275-3}}.

\bibitem{Downey1995}
Rodney~G. Downey and Michael~R. Fellows.
\newblock Fixed-parameter tractability and completeness {II:} on completeness
  for {W[1]}.
\newblock {\em Theor. Comput. Sci.}, 141(1{\&}2):109--131, 1995.
\newblock \href {https://doi.org/10.1016/0304-3975(94)00097-3}
  {\path{doi:10.1016/0304-3975(94)00097-3}}.

\bibitem{Dreze1980}
Jacques~H. Drèze and Joseph Greenberg.
\newblock Hedonic coalitions: Optimality and stability.
\newblock {\em Econometrica}, 48(4):987--1003, 1980.
\newblock URL: \url{http://www.jstor.org/stable/1912943}.

\bibitem{edmonds1965maximum}
Jack Edmonds.
\newblock Maximum matching and a polyhedron with 0, 1-vertices.
\newblock {\em Journal of Research of the National Bureau of Standards B},
  69(125-130):55--56, 1965.

\bibitem{Gabow17}
Harold~N. Gabow.
\newblock A data structure for nearest common ancestors with linking.
\newblock {\em {ACM} Trans. Algorithms}, 13(4):45:1--45:28, 2017.
\newblock \href {https://doi.org/10.1145/3108240} {\path{doi:10.1145/3108240}}.

\bibitem{Gabow1988}
Harold~N. Gabow and Robert~Endre Tarjan.
\newblock Algorithms for two bottleneck optimization problems.
\newblock {\em J. Algorithms}, 9(3):411--417, 1988.
\newblock \href {https://doi.org/10.1016/0196-6774(88)90031-4}
  {\path{doi:10.1016/0196-6774(88)90031-4}}.

\bibitem{GairingS19}
Martin Gairing and Rahul Savani.
\newblock Computing stable outcomes in symmetric additively separable hedonic
  games.
\newblock {\em Math. Oper. Res.}, 44(3):1101--1121, 2019.
\newblock \href {https://doi.org/10.1287/MOOR.2018.0960}
  {\path{doi:10.1287/MOOR.2018.0960}}.

\bibitem{GS1962}
David Gale and Lloyd~S. Shapley.
\newblock College admissions and the stability of marriage.
\newblock {\em The American Mathematical Monthly}, 69(1):9--15, 1962.
\newblock URL: \url{http://www.jstor.org/stable/2312726}.

\bibitem{GJ1979}
Michael~R. Garey and David~S. Johnson.
\newblock {\em Computers and Intractability: {A} Guide to the Theory of
  NP-Completeness}.
\newblock W. H. Freeman, 1979.

\bibitem{GLW2017}
Laurent Gourv{\`{e}}s, Julien Lesca, and Ana{\"{e}}lle Wilczynski.
\newblock Object allocation via swaps along a social network.
\newblock In Carles Sierra, editor, {\em Proceedings of the Twenty-Sixth
  International Joint Conference on Artificial Intelligence, {IJCAI} 2017,
  Melbourne, Australia, August 19-25, 2017}, pages 213--219. ijcai.org, 2017.
\newblock \href {https://doi.org/10.24963/IJCAI.2017/31}
  {\path{doi:10.24963/IJCAI.2017/31}}.

\bibitem{GI1989}
Dan Gusfield and Robert~W. Irving.
\newblock {\em The Stable marriage problem - structure and algorithms}.
\newblock Foundations of computing series. {MIT} Press, 1989.

\bibitem{HanakaIO25}
Tesshu Hanaka, Airi Ikeyama, and Hirotaka Ono.
\newblock Maximizing utilitarian and egalitarian welfare of fractional hedonic
  games on tree-like graphs.
\newblock {\em J. Comb. Optim.}, 49(3):53, 2025.
\newblock \href {https://doi.org/10.1007/S10878-025-01283-6}
  {\path{doi:10.1007/S10878-025-01283-6}}.

\bibitem{Hanaka2019}
Tesshu Hanaka, Hironori Kiya, Yasuhide Maei, and Hirotaka Ono.
\newblock Computational complexity of hedonic games on sparse graphs.
\newblock In Matteo Baldoni, Mehdi Dastani, Beishui Liao, Yuko Sakurai, and Rym
  Zalila{-}Wenkstern, editors, {\em Proceedings of the 22nd International
  Conference on Principles and Practice of Multi-Agent Systems, {PRIMA} 2019,
  Turin, Italy, October 28-31, 2019}, volume 11873 of {\em Lecture Notes in
  Computer Science}, pages 576--584. Springer, 2019.
\newblock \href {https://doi.org/10.1007/978-3-030-33792-6_43}
  {\path{doi:10.1007/978-3-030-33792-6_43}}.

\bibitem{HM2022}
Tesshu Hanaka and Michael Lampis.
\newblock Hedonic games and treewidth revisited.
\newblock In Shiri Chechik, Gonzalo Navarro, Eva Rotenberg, and Grzegorz
  Herman, editors, {\em Proceedings of the 30th Annual European Symposium on
  Algorithms, {ESA} 2022, September 5-9, 2022, Berlin/Potsdam, Germany}, volume
  244 of {\em LIPIcs}, pages 64:1--64:16. Schloss Dagstuhl - Leibniz-Zentrum
  f{\"{u}}r Informatik, 2022.
\newblock \href {https://doi.org/10.4230/LIPIcs.ESA.2022.64}
  {\path{doi:10.4230/LIPIcs.ESA.2022.64}}.

\bibitem{Igarashi2016}
Ayumi Igarashi and Edith Elkind.
\newblock Hedonic games with graph-restricted communication.
\newblock In Catholijn~M. Jonker, Stacy Marsella, John Thangarajah, and Karl
  Tuyls, editors, {\em Proceedings of the 15th International Conference on
  Autonomous Agents {\&} Multiagent Systems, Singapore, May 9-13, 2016}, pages
  242--250. {ACM}, 2016.
\newblock URL: \url{http://dl.acm.org/citation.cfm?id=2936962}.

\bibitem{IRVING1985}
Robert~W. Irving.
\newblock An efficient algorithm for the "stable roommates" problem.
\newblock {\em J. Algorithms}, 6(4):577--595, 1985.
\newblock \href {https://doi.org/10.1016/0196-6774(85)90033-1}
  {\path{doi:10.1016/0196-6774(85)90033-1}}.

\bibitem{Irving1994}
Robert~W. Irving.
\newblock Stable marriage and indifference.
\newblock {\em Discret. Appl. Math.}, 48(3):261--272, 1994.
\newblock \href {https://doi.org/10.1016/0166-218X(92)00179-P}
  {\path{doi:10.1016/0166-218X(92)00179-P}}.

\bibitem{Irving2008}
Robert~W. Irving.
\newblock Stable matching problems with exchange restrictions.
\newblock {\em J. Comb. Optim.}, 16(4):344--360, 2008.
\newblock \href {https://doi.org/10.1007/S10878-008-9153-1}
  {\path{doi:10.1007/S10878-008-9153-1}}.

\bibitem{Irving2002}
Robert~W. Irving and David~F. Manlove.
\newblock The stable roommates problem with ties.
\newblock {\em J. Algorithms}, 43(1):85--105, 2002.
\newblock \href {https://doi.org/10.1006/JAGM.2002.1219}
  {\path{doi:10.1006/JAGM.2002.1219}}.

\bibitem{KerkmannNRRRSW22}
Anna~Maria Kerkmann, Nhan{-}Tam Nguyen, Anja Rey, Lisa Rey, J{\"{o}}rg Rothe,
  Lena Schend, and Alessandra Wiechers.
\newblock Altruistic hedonic games.
\newblock {\em J. Artif. Intell. Res.}, 75, 2022.
\newblock \href {https://doi.org/10.1613/JAIR.1.13706}
  {\path{doi:10.1613/JAIR.1.13706}}.

\bibitem{Kijima2012}
Shuji Kijima, Yota Otachi, Toshiki Saitoh, and Takeaki Uno.
\newblock Subgraph isomorphism in graph classes.
\newblock {\em Discret. Math.}, 312(21):3164--3173, 2012.
\newblock \href {https://doi.org/10.1016/J.DISC.2012.07.010}
  {\path{doi:10.1016/J.DISC.2012.07.010}}.

\bibitem{KreiselBFN24}
Luca Kreisel, Niclas Boehmer, Vincent Froese, and Rolf Niedermeier.
\newblock Equilibria in schelling games: computational hardness and robustness.
\newblock {\em Auton. Agents Multi Agent Syst.}, 38(1):9, 2024.
\newblock \href {https://doi.org/10.1007/S10458-023-09632-7}
  {\path{doi:10.1007/S10458-023-09632-7}}.

\bibitem{Manlove2013}
David~F. Manlove.
\newblock {\em Algorithmics of Matching Under Preferences}, volume~2 of {\em
  Series on Theoretical Computer Science}.
\newblock WorldScientific, 2013.
\newblock \href {https://doi.org/10.1142/8591} {\path{doi:10.1142/8591}}.

\bibitem{Massand2019}
Sagar Massand and Sunil Simon.
\newblock Graphical one-sided markets.
\newblock In Sarit Kraus, editor, {\em Proceedings of the Twenty-Eighth
  International Joint Conference on Artificial Intelligence, {IJCAI} 2019,
  Macao, China, August 10-16, 2019}, pages 492--498. ijcai.org, 2019.
\newblock \href {https://doi.org/10.24963/IJCAI.2019/70}
  {\path{doi:10.24963/IJCAI.2019/70}}.

\bibitem{Noam2007}
Noam Nisan, Tim Roughgarden, {\'{E}}va Tardos, and Vijay~V. Vazirani, editors.
\newblock {\em Algorithmic Game Theory}.
\newblock Cambridge University Press, 2007.
\newblock \href {https://doi.org/10.1017/CBO9780511800481}
  {\path{doi:10.1017/CBO9780511800481}}.

\bibitem{Okubo2019}
Masahiro Okubo, Tesshu Hanaka, and Hirotaka Ono.
\newblock Optimal partition of a tree with social distance.
\newblock In Gautam~K. Das, Partha~Sarathi Mandal, Krishnendu Mukhopadhyaya,
  and Shin{-}Ichi Nakano, editors, {\em Proceedings of the 13th International
  Conference on Algorithms and Computation, {WALCOM} 2019, Guwahati, India,
  February 27 - March 2, 2019}, volume 11355 of {\em Lecture Notes in Computer
  Science}, pages 121--132. Springer, 2019.
\newblock \href {https://doi.org/10.1007/978-3-030-10564-8_10}
  {\path{doi:10.1007/978-3-030-10564-8_10}}.

\bibitem{OlsenBT12}
Martin Olsen, Lars B{\ae}kgaard, and Torben Tambo.
\newblock On non-trivial nash stable partitions in additive hedonic games with
  symmetric 0/1-utilities.
\newblock {\em Inf. Process. Lett.}, 112(23):903--907, 2012.
\newblock \href {https://doi.org/10.1016/J.IPL.2012.08.016}
  {\path{doi:10.1016/J.IPL.2012.08.016}}.

\bibitem{Peters2016}
Dominik Peters.
\newblock Graphical hedonic games of bounded treewidth.
\newblock In Dale Schuurmans and Michael~P. Wellman, editors, {\em Proceedings
  of the Thirtieth {AAAI} Conference on Artificial Intelligence, February
  12-17, 2016, Phoenix, Arizona, {USA}}, pages 586--593. {AAAI} Press, 2016.
\newblock \href {https://doi.org/10.1609/AAAI.V30I1.10046}
  {\path{doi:10.1609/AAAI.V30I1.10046}}.

\bibitem{Saffidine2018}
Abdallah Saffidine and Ana{\"{e}}lle Wilczynski.
\newblock Constrained swap dynamics over a social network in distributed
  resource reallocation.
\newblock In Xiaotie Deng, editor, {\em Proceedings of the 11th International
  Symposium on Algorithmic Game Theory, {SAGT} 2018, Beijing, China, September
  11-14, 2018}, volume 11059 of {\em Lecture Notes in Computer Science}, pages
  213--225. Springer, 2018.
\newblock \href {https://doi.org/10.1007/978-3-319-99660-8_19}
  {\path{doi:10.1007/978-3-319-99660-8_19}}.

\bibitem{SG:Schelling1969}
Thomas~C. Schelling.
\newblock Models of segregation.
\newblock {\em The American Economic Review}, 59(2):488--493, 1969.
\newblock URL: \url{http://www.jstor.org/stable/1823701}.

\bibitem{SG:Schelling1971}
Thomas~C. Schelling.
\newblock Dynamic models of segregation†.
\newblock {\em The Journal of Mathematical Sociology}, 1(2):143--186, 07 1971.
\newblock \href {https://doi.org/10.1080/0022250X.1971.9989794}
  {\path{doi:10.1080/0022250X.1971.9989794}}.

\bibitem{SA:Wilczynski23}
Ana{\"{e}}lle Wilczynski.
\newblock Ordinal hedonic seat arrangement under restricted preference domains:
  Swap stability and popularity.
\newblock In {\em Proceedings of the Thirty-Second International Joint
  Conference on Artificial Intelligence, {IJCAI} 2023, 19th-25th August 2023,
  Macao, SAR, China}, pages 2906--2914. ijcai.org, 2023.
\newblock \href {https://doi.org/10.24963/IJCAI.2023/324}
  {\path{doi:10.24963/IJCAI.2023/324}}.

\end{thebibliography}


\end{document}